\documentclass[lettersize,journal]{IEEEtran}
\usepackage{amsmath,amsfonts}
\usepackage{algorithmic}
\usepackage{algorithm}
\usepackage{array}
\usepackage[caption=false,font=normalsize,labelfont=sf,textfont=sf]{subfig}
\usepackage{textcomp}
\usepackage{stfloats}
\usepackage{url}
\usepackage{verbatim}
\usepackage{graphicx}
\usepackage{cite}
\usepackage{mathrsfs}
\usepackage{color}
\usepackage{bm}
\usepackage{amssymb}
\usepackage{amsthm}
\usepackage{caption}
\usepackage{epstopdf}
\usepackage{booktabs}
\usepackage{epsfig}
\usepackage{float}
\usepackage{xcolor}  

\usepackage{hyperref}
\usepackage{todonotes}

\newtheorem{thm}{Theorem}

\newtheorem{remark}{Remark}
\newtheorem{lem}{Lemma}

\newtheorem{example}{Example}
\newtheorem{proposition}{Proposition}

\DeclareMathOperator*{\argmin}{arg\,min}
	
\begin{document}
		
		\title{Reconstruction of Graph Signals on Complex Manifolds with Kernel Methods}
		
		\author{Yu Zhang, Linyu Peng,~\IEEEmembership{Member,~IEEE,} and Bing-Zhao Li$^{\ast}$,~\IEEEmembership{Member,~IEEE,}
			\thanks{The current study was partially supported by the National Natural Science Foundation of China (No. 62171041), Natural Science Foundation of Beijing Municipality (No. 4242011), JSPS KAKENHI (24K06852), JST CREST (JPMJCR24Q5), and Keio University (Academic Development Fund, Fukuzawa Fund). \textit{(Corresponding author: Bing-Zhao Li.)}}
			\thanks{Yu Zhang is with the School of Mathematics and Statistics, Beijing Institute of Technology, Beijing 102488, China, and also with the Department of Mechanical Engineering, Keio University, Yokohama 223-8522, Japan (e-mail: zhangyu$\_$bit@keio.jp).}
			\thanks{Linyu Peng is with the Department of Mechanical Engineering, Keio University, Yokohama 223-8522, Japan (e-mail: l.peng@mech.keio.ac.jp).}
			\thanks{Bing-Zhao Li is with the School of Mathematics and Statistics, Beijing Institute of Technology, Beijing 102488, China (e-mail: li$\_$bingzhao@bit.edu.cn).}}
		
		\markboth{Journal of \LaTeX\ Class Files,~Vol.~, No.~, ~2024}%
		{Shell \MakeLowercase{\textit{et al.}}: A Sample Article Using IEEEtran.cls for IEEE Journals}
		
		\IEEEpubid{0000--0000/00\$00.00~\copyright~2024 IEEE}
		
		\maketitle
		
		
		\begin{abstract}
			Graph signals are widely used to describe vertex attributes or features in graph-structured data, with applications spanning the internet, social media, transportation, sensor networks, and biomedicine. Graph signal processing (GSP) has emerged to facilitate the analysis, processing, and sampling of such signals.  While kernel methods have been extensively studied for estimating graph signals from samples provided on a subset of vertices, their application to complex-valued graph signals remains largely unexplored. This paper introduces a novel framework for reconstructing graph signals using kernel methods on complex manifolds. By embedding graph vertices into a higher-dimensional complex ambient space that approximates a lower-dimensional manifold, the framework extends the reproducing kernel Hilbert space to complex manifolds. It leverages Hermitian metrics and geometric measures to characterize kernels and graph signals. Additionally, several traditional kernels and graph topology-driven kernels are proposed for reconstructing complex graph signals. Finally, experimental results on synthetic and real-world datasets demonstrate the effectiveness of this framework in accurately reconstructing complex graph signals, outperforming conventional kernel-based approaches. This work lays a foundational basis for integrating complex geometry and kernel methods in GSP.
		\end{abstract}
		
		
		
		\section{Introduction}
		\label{Intro}
		\IEEEPARstart{I}{n} real-world signal processing, data in fields such as social media, communications, sensors, transportation, and biomedical applications are often associated with networks, and their irregular vertex attributes or features can be interpreted as graph functions or graph signals \cite{Network,Kernel,GFTlaplace}. Graph signal processing (GSP) has emerged as a crucial field for addressing the challenges of analyzing, processing, and sampling data defined on graph or network structures \cite{GFTlaplace,GFTadjacency2,Goverview,Ghistory}. GSP characterizes and analyzes the relationship between signals and graphs through the graph Fourier transform (GFT) and the corresponding spectral domain. In practical applications, GSP has been applied to image denoising \cite{Gdenoising}, brain signal analysis \cite{Gbrain},  the design of recommendation systems \cite{Grecommendation}, to mention only a few.
		
		The vast number of nodes in processing, transportation, and storage networks often requires significant computational and storage resources. Graph signal sampling addresses this issue by reducing the number of samples on the graph while enabling reconstruction of the underlying signal. It extends the conventional sampling paradigm to graph signals. The primary goal of graph signal reconstruction is to recover the entire signal based on observations from a subset of vertices. The main challenge in graph signal reconstruction is the design of optimal sampling and recovery strategies \cite{Gsampling}. Most existing methods focus on smooth or bandlimited signals, where the energy is typically concentrated in the eigenvector subspace of the Laplacian or adjacency matrix \cite{GUncertainty,GFTsampling,GFTSSS,GFTdualizing,GLCTsampling,Ggeneralizedsamp,GParallel}. As a result, sampling becomes an optimization problem of selecting the best rows of the eigenvector matrix corresponding to the GFT \cite{GUncertainty,GFTsampling,GFTSSS,GFTdualizing,Gfrequency}. However, these methods rely on the spectral decomposition of the Laplacian or adjacency matrix, which is computationally expensive, typically requiring $\mathcal{O}(N^3)$ operations with $N$ dimension of the matrix \cite{Ginterpolation}. 
		\IEEEpubidadjcol
		
		Kernel-based methods have gained widespread recognition for their effectiveness in reconstructing graph signals from limited samples available on a subset of vertices. Their flexibility and relatively low computational cost have driven significant advancements in graph signal filtering and reconstruction. Using the smoothness of the graph signals and the underlying graph topology, kernel-based reconstruction has been extensively studied in the context of graph signals. Early studies \cite{Kernel,MKL} introduced kernel regression as a framework to generalize traditional GSP methods, employing reproducing kernel Hilbert space (RKHS) and multi-kernel learning (MKL) strategies for flexible signal reconstruction. These methods were further extended to space-time signal reconstruction in dynamic graphs \cite{Kerneldynamic} and hybrid parametric and non-parametric models \cite{Semi-kernel}. A fast sampling method based on an extended sampling theorem was proposed under this framework \cite{Extendedsampling}. In \cite{Predictingkernel}, a kernel ridge regression (KRR) method was introduced, framing the learning problem for graph signals, while positive definite functions were employed for graph signal interpolation \cite{PDkernel}. More recently, \cite{Ginterpolation} developed a Gaussian mixture-based RKHS for graph signal interpolation and extrapolation. Furthermore, \cite{Generalizedkernel} applied kernel methods to signal reconstruction within a generalized GSP framework, while \cite{Bandlimitedkernel} focused on bandlimited signals, effectively recovering them using kernels. Gradient approaches have also been proposed to improve signal reconstruction on manifold graphs \cite{Gradient}. 
		
		These studies leverage kernels to define relationships between data points and capture the underlying signal smoothness and geometry, enabling the approximation of missing values and reconstruction of signals across all vertices, thus providing a powerful framework for graph signal reconstruction. However, while kernel functions are widely used in real-valued signal reconstruction, they are primarily designed for such settings and overlook the critical Hermitian properties of complex-valued signals. The extension of kernel methods to reconstruct complex-valued graph signals remains largely unexplored. Complex-valued signals, prevalent in applications such as electrical circuits, quantum networks, and wireless communications \cite{Complex-Valued}, carry magnitude and phase information, crucial to understanding the underlying physical systems. Traditional graph signal reconstruction methods do not capture phase information and intricate dependencies between real and imaginary components, necessitating specialized techniques to address the unique structure of complex signals. Furthermore, existing methods often lack the ability to integrate rich geometric structures in complex domains, limiting their applicability to more sophisticated signal processing tasks.
		
		Advances in complex manifolds and Hermitian geometry \cite{Complexmanifolds1, Complexmanifolds2, Hermitian} offer a novel pathway for extending kernel methods to the complex domain. Complex manifolds, as a generalization of Riemannian manifold, provide a natural framework for representing complex-valued data with smooth geometric and topological structures. In particular, Hermitian metrics and holomorphic mappings allow for the precise characterization of geometric relationships within complex-valued signal spaces. However, integrating these mathematical tools with GSP still requires further development.
		
		The motivation of this paper stems from the need to address these limitations and develop robust methods for handling complex-valued graph signals. Reconstructing complex-valued signals on graphs presents unique challenges, such as preserving Hermitian properties and maintaining coherence between signal components. Additionally, the geometric and topological richness of complex manifolds naturally aligns with the principles of RKHS \cite{RKHS}, offering a powerful means to extend kernel methods into the complex domain. This work aims to bridge the gap between complex geometry and GSP by developing a unified framework for reconstructing complex graph signals using kernel methods specifically tailored for manifold structures. 
		
		We summarize our contributions as follows:
		\begin{enumerate}
			\item Introduction of a novel framework for reconstructing complex graph signals using kernels on complex manifolds, integrating complex geometry and kernels into GSP.
			\item Extension of the RKHS framework to complex manifolds by introducing Hermitian metrics and geometric measures.  
			\item Development of traditional kernels and graph topology-driven kernels for reconstructing complex-valued graph signals, along with a generalization of MKL methods.  
			\item Validation of the proposed framework on synthetic and real-world datasets, demonstrating performance compared to traditional kernel-based methods on manifolds.  
		\end{enumerate}
		
		The remainder of this paper is structured as follows. Section II reviews graph signals and RKHS, focusing on sampling and reconstruction. Section III introduces complex manifolds, the RKHS on these manifolds, and kernel functions for graph signal reconstruction. Section IV presents Hermitian metrics and geometric measures for traditional and graph topology-driven kernels. Section V extends multidimensional kernel learning methods to improve reconstruction performance. Section VI validates the proposed methods on synthetic and real-world datasets. Finally, Section VII concludes the paper.  
		
		\textit{Notation.}
		Scalars are denoted by lowercase letters, vectors by bold lowercase letters, and matrices by bold uppercase letters. Sets or manifolds are represented by calligraphic letters, with $| \cdot |$ denoting the cardinality of a set. The $(n, m)$-th entry of a matrix $\mathbf{W}$ is $w_{nm}$. The symbols $\| \cdot \|_2$ and $\text{Tr}(\cdot)$ represent the Euclidean norm and the trace, respectively. $\mathbf{I}_N$ denotes the $N \times N$ identity matrix, abbreviated as $\mathbf{I}$ when the dimension is obvious. $\mathbf{0}$ (or $\mathbf{1})$ refers to a vector of appropriate dimension with all zeros (or ones). The span of the matrix columns is denoted by $\operatorname{span_c}\{\cdot\}$, and $\delta$ represents the Kronecker delta function. 
		The imaginary unit is denoted by $\mathrm{i}$. The superscripts $\top$, $\mathrm{H}$, and $\dag$ indicate the transpose, conjugate transpose, and pseudo-inverse, respectively, while $\overline{(\cdot)}$ denotes the complex conjugate. 
		
		
		\section{Preliminaries}
		\label{Preliminaries}
		\subsection{Graph Signals}
		An undirected weighted graph $\mathcal{G} = (\mathcal{V}, \mathcal{E}, \mathbf{W})$ is defined by its vertex set $\mathcal{V} = \{v_0, v_1, \dots, v_{N-1}\}$ and edge set $\mathcal{E}$ of size $|\mathcal{V}| = N$. A graph signal is defined as a function $f$ from  $\mathcal{V}$ to  $\mathbb{R}$ or $\mathbb{C}$, typically represented as a vector $\bm{f}$ in $\mathbb{R}^N$ or $\mathbb{C}^N$, respectively. The weighted adjacency matrix $\mathbf{W} \in \mathbb{R}^{N \times N}$ is symmetric and encodes the weights of the edges and the pairwise relationships between the graph vertices. If there is an edge $(n, m)$ connecting the vertices $v_n$ and $v_m$, the weight of this edge is denoted by $w_{nm}$; otherwise, $w_{nm}=0$ \cite{GFTlaplace}.
		
		In many cases, edge weights are not defined naturally. In practice, however, $\mathbf{W}$ can be generated through an application-dependent kernel function $\kappa: \mathcal{V} \times \mathcal{V} \to \mathbb{R}$, which captures the pairwise similarities between the nodes. Specifically, $v_n$ and $v_m$ may be associated with feature vectors, and the corresponding kernels are assumed to be symmetric positive definite (SPD) \cite{Kernel}. The resulting weighted adjacency matrix can serve as a graph shift operator, allowing the GFT to be defined through the eigenvalues and eigenvectors obtained from its matrix decomposition \cite{GFTadjacency2}.
		
		Alternatively, the graph signal can also be defined by eigendecomposition of a Laplacian matrix \cite{GFTlaplace,GFTadjacency2,Goverview}. 
		
		For a symmetric difference operator, the non-normalized graph Laplacian is given by $\mathbf{L} = \mathbf{D}^{\mathrm{deg}} - \mathbf{W}$, where $\mathbf{D}^{\mathrm{deg}}= \text{diag}(\mathbf{W}\mathbf{1})$ denotes the degree matrix of the graph $\mathcal{G}$. The vector $\mathbf{W} \mathbf{1}$ represents the row-wise summation of the adjacency matrix $\mathbf{W}$, with each entry corresponding to the degree of the respective vertex in the graph. The normalized graph Laplacian is expressed as $$\mathbf{L}_{\text{norm}} = \mathbf{I} - (\mathbf{D}^{\mathrm{deg}})^{-1/2} \mathbf{W} (\mathbf{D}^{\mathrm{deg}})^{-1/2}.$$

		\begin{example}
			We simulate two commonly used graphs embedded within manifolds: a $1000$-node two moons graph and a $1500$-node Swiss roll graph (Fig. \ref{fig01}). In both cases, the signals are random.
		\end{example}
		
		\begin{figure}[htbp]
			\begin{center}
				\begin{minipage}[t]{0.45\linewidth}
					\centering
					\includegraphics[width=\linewidth]{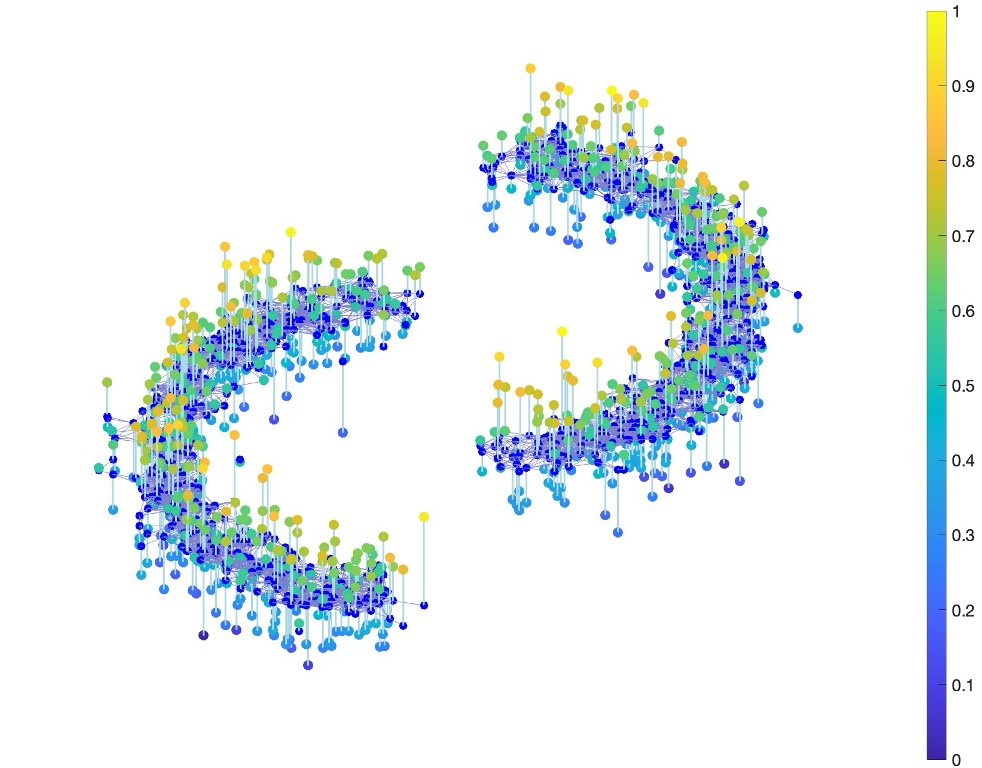}
					\parbox{1.7cm}{\tiny (a) Two moons graph.}
				\end{minipage}
				\begin{minipage}[t]{0.45\linewidth}
					\centering
					\includegraphics[width=\linewidth]{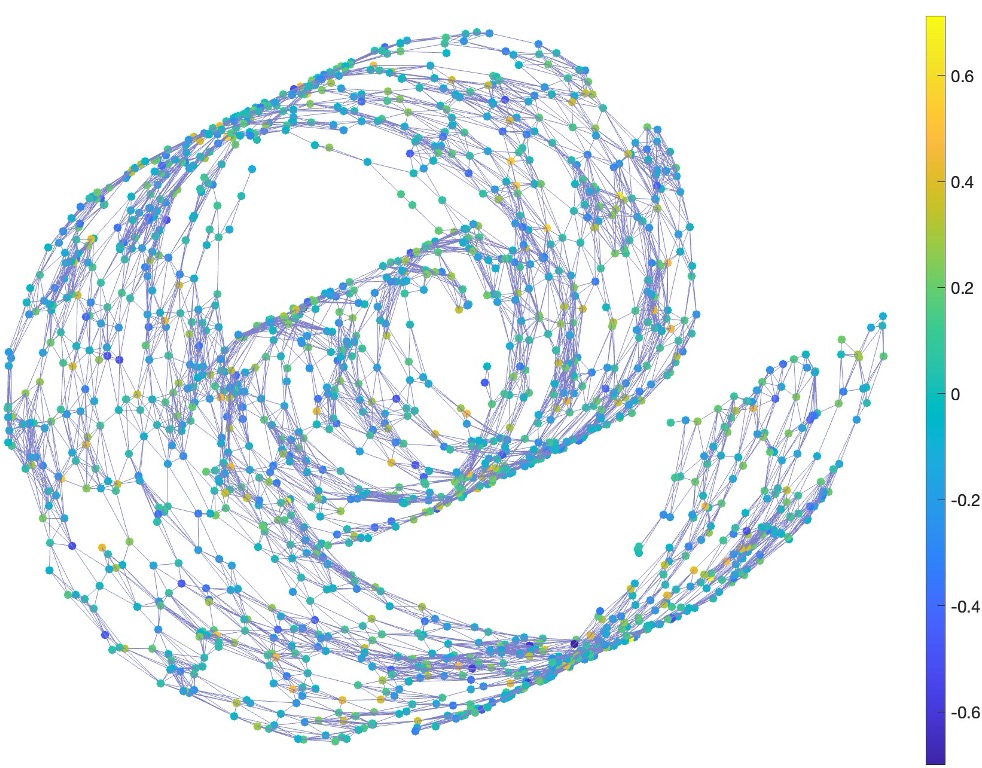}
					\parbox{1.7cm}{\tiny (b) Swiss roll graph.}
				\end{minipage}
			\end{center}
			\caption{Random signals on graphs: (a) Two nested semicircles with light blue columns and color-coded endpoints representing signal values. (b) A 3D spiral manifold with vertex colors indicating signal values. Edges are determined by adjacency matrices constructed using a Gaussian kernel \cite{GSPBOX}.}
			\label{fig01}
		\end{figure}
		
		
		\subsection{The Reproducing Kernel Hilbert Space}
		In GSP with real-valued graph signals, SPD kernels allow the introduction of inner products in Hilbert spaces. For any SPD kernel $\kappa$, consider the points $x, x' \in \mathcal{X}$, where $\mathcal{X}$ is a simple differentiable manifold, the following inequality holds \cite{RKHS}  
		\[
		\kappa(x, x')^2 \leq \kappa(x, x) \kappa(x', x'),
		\]  
		which ensures the existence of a Hilbert space $\mathcal{H}$ and a feature mapping $ \Phi: \mathcal{X} \to \mathcal{H}$ such that for all $x, x' \in \mathcal{X}$,  
		\[
		\kappa(x, x') = \langle \Phi(x), \Phi(x') \rangle.
		\]  
		Furthermore, $\mathcal{H}$ satisfies the reproducing property, which states that for any ${f} \in \mathcal{H}$ and $x \in \mathcal{X}$,  
		\[
		{f}(x) = \langle {f}, \kappa(x, \cdot) \rangle.
		\]  
		As a result, $\mathcal{H}$ is referred to as the RKHS associated with $\kappa$.  
		In practical implementations, the kernel function $\kappa(x, x')$ is evaluated for a finite set of vertices $\mathcal{V}$ on the graph. For any $\{x_1,\ldots, x_M\}$, the evaluations define the kernel matrix $\mathbf{K} \in \mathbb{R}^{M \times M}$ as  
		\[
		\mathbf{K}_{ij} = \kappa(x_i, x_j),
		\]
		where $x_i$ and $x_j$ represent the nodes corresponding to the $i$-th and $j$-th vertices in the set $\mathcal{V}$. The kernel matrix $\mathbf{K}$ is symmetric positive semi-definite. It encapsulates pairwise relationships between graph signals and provides a structured representation of the graph signal space in terms of inner products in $\mathcal{H}$ \cite{Hilbert}.  
		
		The Hilbert space $\mathcal{H}$ associated with an SPD kernel $\kappa$ is a structured RKHS. For any Hilbert space $\mathcal{H}$ admitting a feature mapping $\Phi: \mathcal{X} \to \mathcal{H}$ such that $\kappa(x, x') = \langle \Phi(x), \Phi(x') \rangle$ for all $x, x' \in \mathcal{X}$, $\mathcal{H}$ serves as the \textit{feature space} associated with $\kappa$, and $\Phi$ is the corresponding feature mapping.  
		
		The norm in the feature space $\mathcal{H}$, induced by its inner product, is denoted by $ ||\cdot||_{\mathcal{H}}$. For any $f \in \mathcal{H}$, 
		\[
		||f||_{\mathcal{H}} = \sqrt{\langle f, f \rangle}.
		\]  
		It is crucial to note that the feature space associated with a kernel $\kappa$ is not necessarily unique and may vary in dimensionality. The dimension of the feature space can refer  to either the explicit dimension determined by a specific feature mapping or the dimensionality of the RKHS associated with $\kappa$.
		
		\subsection{Sampling and Reconstruction of Graph Signals in RKHS}  
		In scenarios with limited measurement resources, a sampling set $\mathcal{S} \subseteq \mathcal{V}$ is selected from the graph signal $f \in \mathcal{H}$. The  values of $f$ at all vertices collectively form the graph signal vector $\bm{f} = \left( f(v_1), f(v_2), \dots, f(v_N) \right)^{\top}$, which is assumed to be real-valued in this subsection.   The sampling process is represented by a sampling matrix $\mathbf{D} \in \{0, 1\}^{|\mathcal{S}| \times N}$, defined as \cite{GFTsampling}  
		\begin{equation}\label{D}
			d_{ij} = 
			\begin{cases}
				1, & \text{if } j = \mathcal{S}\{i\}, \\
				0, & \text{else}.
			\end{cases}
		\end{equation}  
		Here, $|\mathcal{S}|$ denotes the size of the sample set $\mathcal{S}$.
		Suppose a series of noisy graph samples (or observations) $\bm{y}_{\mathcal{S}}$ are available, which is modeled  as 
		\begin{equation}
			\bm{y}_{\mathcal{S}} = \mathbf{D}\bm{f} + \bm{n}, \label{yS}
		\end{equation}  
		where $\bm{n}$ is a Gaussian noise.
		
		The objective of graph signal reconstruction is to estimate the original function $f$ by solving an optimization problem that minimizes the squared error between the measured noisy signal and the reconstructed signal \cite{Extendedsampling}, combined with a regularization term that penalizes large norms in the RKHS $ \mathcal{H} $ as follows  
		\begin{equation}
			f^{\text{opt}} = \argmin_{f \in \mathcal{H}} \left(\frac{1}{|\mathcal{S}|} \left\| \bm{y}_{\mathcal{S}} - \bm{f}_\mathcal{S} \right\|^2 + \gamma \|f\|_{\mathcal{H}}^2\right), \label{fopt}
		\end{equation} 
		where $\bm{f}_\mathcal{S} = \mathbf{D}\bm{f}$ represents the reconstructed signal vector $\bm{f}$ on the sampled vertices in the subset $\mathcal{S}$, and $ \gamma$ is the regularization parameter balancing reconstruction accuracy and smoothness in the feature space.
		
		Alternatively, by employing a kernel-based approach that relates the graph signal $\bm{f}$ to the kernel matrix via $\bm{f} = \mathbf{K} \bm{\alpha}$ \cite{Kernel,Learningkernel}, the kernel coefficients $\bm{\alpha}$ can be estimated by solving the following optimization problem  
		\[
		\bm{\alpha}^{\text{opt}} = \argmin_{\bm{\alpha} \in \mathbb{R}^N} \left(\frac{1}{|\mathcal{S}|} \left| \left| \bm{y}_\mathcal{S} - \mathbf{DK} \bm{\alpha} \right| \right|^2 + \gamma \bm{\alpha}^\top \mathbf{K} \bm{\alpha}\right).
		\]  
		
		This can be equivalently  reformulated as 
		\[
		\bm{f}^{\text{opt}} = \argmin_{\bm{f} \in \operatorname{span_c}\{\mathbf{K}\}} \left(\frac{1}{|\mathcal{S}|} \left| \left| \bm{y}_{\mathcal{S}} - \mathbf{D}\bm{f} \right| \right|^2 + \gamma \bm{f}^\top \mathbf{K}^\dagger \bm{f}\right),
		\] 
		which represents the Tikhonov regularization \cite{TR} for graph signal reconstruction. If $ \mathbf{K}$ contains null eigenvalues, it is typically perturbed as $\mathbf{K} + \epsilon \mathbf{I}$ to ensure invertibility. The solution can be efficiently obtained using KRR, where the original graph signal is interpolated as a linear combination of the sampled values.

		
		\section{Graph Signal Reconstruction with Kernels over Complex Manifolds}
		\label{comlpex}
		In many applications, particularly in the fields of machine learning and signal processing, kernel methods have emerged as a dominant approach for nonlinear function estimation \cite{Learningkernel}. Their popularity can be attributed to their simplicity, flexibility, and excellent performance. Furthermore, kernel methods find extensive applications in complex manifolds. This section introduces kernels on complex manifolds \cite{Complexmanifolds2} as a novel framework for graph signal reconstruction and explores the implications of representation theorems.
		
		To introduce kernel methods, we first present an assumption that constructs a complex manifold $\mathcal{M}$ and embeds the vertices of the graph $\mathcal{G}$ into it. A complex manifold is a topological space that can be locally mapped to the complex space $\mathbb{C}^N$, the transition maps between these local charts being holomorphic, thereby ensuring an analytic structure.
		
		
		\subsection{Complex Manifold Assumption}
		In our assumption of the complex manifold model, we typically posit that the actual data resides within a higher-dimensional complex ambient space $\mathcal{Z}$ and is embedded in or approximates a lower-dimensional complex manifold $\mathcal{M} \subset \mathcal{Z}$. The motivation behind this assumption is that, despite the higher-dimensional and complex manifestation of the data, its intrinsic structure can often be effectively described by a lower-dimensional complex geometric form \cite{Ginterpolation}. Assuming the ambient space $\mathcal{Z}$ is equipped with a probability distribution $P_{\mathcal{Z}}$, for a point $z\in \mathcal{Z}$ (with components $z^{i}$), the complex manifold $\mathcal{M}$ is defined as the support of this distribution, reading 
		\[
		\mathcal{M} = \text{supp}(P_{\mathcal{Z}}) = \{ z \in \mathcal{Z} \mid P_{\mathcal{Z}}(z) \neq 0 \} \subset \mathcal{Z},
		\]  
		which satisfies
		\[
		P_{\mathcal{Z}}(\mathcal{M}) = \int_{\mathcal{M}} p(z) \, \mathrm{vol}_{\mathcal{M}}(z) = 1,
		\]  
		where $p(z)$ is the probability density function, and $\mathrm{d}P_{\mathcal{Z}}(z) = p(z) \, \mathrm{vol}_{\mathcal{M}}(z)$. The measure is induced by a Hermitian metric $h$ \cite{Hermitian}.  
		
		The Hermitian metric $h$ generalizes the notion of distance on $\mathcal{M}$, reflecting both the local geometry and the holomorphic structure of the manifold. It is defined on the holomorphic tangent bundle, given in local coordinates $\{z^{i}\}$ as
		\[
		h = h_{i \bar{j}} \mathrm{d}z^{i} \otimes \mathrm{d}\bar{z}^{j},
		\]
		where \(h_{i \bar{j}}\) is the local metric tensor. This metric naturally induces a volume form on \(\mathcal{M}\), 
		\[
		\operatorname{vol}_{\mathcal{M}} = \left(\frac{\mathrm{i}}{2}\right)^{n}\det\left(\mathbf{H}\right) \, \mathrm{d}z^{1} \wedge \mathrm{d}\bar{z}^{1} \wedge \dots \wedge \mathrm{d}z^{n} \wedge \mathrm{d}\bar{z}^{n},
		\]
		where $\mathbf{H} = (h_{i\bar{j}})$ is the matrix form of the metric tensor. 
		
		For applications involving complex signals, the probability distribution $p(z)$ should align with the geometry of the manifold, capturing signal-specific properties while maintaining holomorphic consistency. Further details of Hermitian metrics and their role in defining the holomorphic structure are briefly reviewed in \textbf{Appendix} \ref{appendix A}.
		
		Through this complex manifold assumption, we infer that the distribution of data in the higher-dimensional complex domain reflects an intrinsic lower-dimensional complex geometric structure, and these distributions are closely tied to the geometric attributes and probability densities on the manifold. This assumption provides us with geometric tools to analyze complex graph signals and their embedded structures.
		
		
		\subsection{RKHS in Complex Manifolds}
		As the counterpart to the real-valued case, a kernel $\kappa: \mathcal{Z} \times \mathcal{Z} \rightarrow \mathbb{C}$ defined in a complex manifold $\mathcal{Z}$ will be assumed to be Hermitian\footnote{A kernel $\kappa(z, z')$ is Hermitian if $\kappa(z, z') = \overline{\kappa(z', z)}$, which reduces to standard symmetry on real-valued manifolds.} positive semi-definite, and will be referred to simply as a positive kernel. In the current study, $\mathcal{Z}$ is the ambient space that contains the manifold $\mathcal{M}$. An RKHS is a Hilbert space of functions defined in $\mathcal{Z}$, which is associated with the kernel $\kappa$. We introduce an alternative definition that emphasizes the one-to-one correspondence between the kernel  $\kappa$ and its RKHS $\mathcal{H}$.
		
		This alternative definition is derived from the space $\mathcal{L}^2_p(\mathcal{Z})$, the space of square-integrable complex functions with respect to the probability measure $P_{\mathcal{Z}}$, and an integral operator associated with the kernel $\kappa$ \cite{Ginterpolation}. The inner product of two functions $f, g \in \mathcal{L}^2_p(\mathcal{Z})$ is given by
		\begin{equation}
			\langle f, g \rangle_{\mathcal{L}^2_p} = \int_{\mathcal{Z}} f(z) \overline{g(z)} p(z) \, \mathrm{vol}_{\mathcal{M}}(z). \label{innerproduct1}
		\end{equation}
		
		The integral operator $L_\kappa: \mathcal{L}^2_p(\mathcal{Z}) \rightarrow \mathcal{L}^2_p(\mathcal{Z})$, associated with the kernel $\kappa$, is defined by
		\begin{equation}
			(L_\kappa f)(z) = \int_{\mathcal{Z}} \kappa(z, z') f(z') p(z') \, \mathrm{vol}_{\mathcal{M}}(z').
		\end{equation}
		It can be shown that, if the kernel is positive, the eigenfunctions of $L_\kappa$ form an orthonormal basis for the RKHS and exhibit a discrete spectrum of positive eigenvalues that decay to zero. This leads to an alternative definition of the RKHS, i.e.,
		\begin{equation}
			\mathcal{H} = \left\{ f: \mathcal{M} \rightarrow \mathbb{C} \, \middle| \, f(z) = \sum_{n=1}^{N} \alpha_n \kappa(z, z_n), \, \alpha_n \in \mathbb{C} \right\},
		\end{equation}
		where $z_n$ represents the node on the manifold $\mathcal{M}$, corresponding to the vertex $v_n$ in the vertex domain $\mathcal{V}$. 
		
		For points $z$ and $z'$, the kernel function $\kappa(z, z')$ measures the similarity between the estimated values $\Phi(z)$ and $\Phi(z')$ in $\mathcal{H}$, where $\Phi$ denotes the feature mapping. If a feature vector associated with point $z$ is available, it is denoted as $\mathbf{z}=\left(\mathrm{z}^1,\mathrm{z}^2,\ldots,\mathrm{z}^D\right) \in \mathbb{C}^{D}$. This definition induces a Hermitian positive semi-definite $N \times N$ matrix, with entries $\mathbf{K}_{nm} := \kappa(z_n, z_m)$.  
		
		When feature vectors $\mathbf{z}_n, \mathbf{z}_m \in \mathbb{C}^D$ associated with graph vertices $\mathcal{V}$ are available, kernels can be naturally defined. For example, the Gaussian kernel \cite{Learningkernel} now can be extended as follows
		\[
		\kappa(z_n, z_{m}) = \exp\left(-\frac{\|\mathbf{z}_n - \mathbf{z}_{m}\|^2}{2\sigma^2}\right), \quad \sigma^2 > 0.
		\]  
		When feature vectors are unavailable, graph kernels can instead be constructed based on the graph topology; see \cite{GFTlaplace} for instance.
		
		Signals in the RKHS are finite-dimensional, as their expansion involves a finite number of terms. This contrasts with more general RKHSs defined on infinite sets, such as $\mathbb{C}^p$ \cite{Kernel}, which are typically infinite-dimensional.
		
		As a result, any function $f \in \mathcal{H}$ can be evaluated on the discrete set of points $\{z_1, z_2, \ldots, z_N\}$, forming a vector $\bm{f} = \left( f(z_1), f(z_2), \ldots, f(z_N)\right)^{\top} \in \mathbb{C}^N$, which can be written as
		\begin{equation}
			\bm{f} = \mathbf{K} \bm{\alpha},\label{Kalpha}
		\end{equation}
		where $\bm{\alpha} = \left(\alpha_1, \alpha_2, \ldots, \alpha_N\right)^{\top}$ is an $N$-dimensional  coefficient vector. The inner product in the RKHS for two functions \begin{equation}
			f(z) = \sum_{n=1}^{N} \alpha_n \kappa(z, z_n),\quad g(z) = \sum_{n=1}^{N} \alpha'_n \kappa(z, z_n)
		\end{equation}  
		is defined by
		\begin{equation}
			\langle f, g \rangle_{\mathcal{H}} =\sum_{n=1}^{N} \sum_{m=1}^{N} \overline{\alpha_n} \alpha'_m \kappa(z_n, z_m) = \bm{\alpha}^{\mathrm{H}} \mathbf{K} \bm{\alpha}',
		\end{equation}
		where $\bm{\alpha}' = \left(\alpha'_1, \alpha'_2, \ldots, \alpha'_N\right)^{\top}$. The corresponding RKHS norm is given  by
		\begin{equation}
			\|f\|^2_{\mathcal{H}} := \langle f, f \rangle_{\mathcal{H}} = \bm{\alpha}^{\mathrm{H}} \mathbf{K} \bm{\alpha},
		\end{equation}
		which can serve as a regularization term to control overfitting. When $\mathbf{K = I}$, the RKHS norm $\|f\|^2_{\mathcal{H}} = \|\bm{f}\|^2_2$ is  equivalent to the Euclidean norm. When the kernel matrix $\mathbf{K}$ is strictly positive definite, the RKHS induced by the kernel bijectively corresponds to the function space $\mathbb{C}^N$. While the set of functions remains the same, the RKHS inner product and norm depend on the specific kernel. This highlights the universality of positive-definite kernels for graph signal reconstruction, with variations primarily reflected in the induced geometry.
		
		The term reproducing kernel arises from the reproducing property. Define the feature mapping $\Psi: \mathcal{M} \to \mathcal{H}$  as 
		\[
		\Psi(z_n) = \kappa(\cdot, z_n), \quad n \in \{1, 2, \dots, N\}.
		\]  
		Then, using the reproducing property, we have  
		\[
		\langle \Psi(z_n), \Psi(z_{m}) \rangle_{\mathcal{H}} = \mathbf{e}_{n}^{\mathrm{H}} \mathbf{K} \mathbf{e}_{m} = \kappa(z_n, z_{m}),
		\]  
		where $n,m \in \{1, 2, \dots, N\}$, and $\mathbf{e}_n$ and $\mathbf{e}_m$ are the $n$-th and $m$-th columns of the identity matrix $\mathbf{I}$, respectively. This expression extracts the corresponding entry from $\mathbf{K}$.
		
		This property is crucial when dealing with RKHS defined on infinite spaces, as it provides an efficient alternative to the expensive multidimensional integrals required for inner products, such as Eq. \eqref{innerproduct1}. 
		The reproducing property simplifies the calculation of the inner product by using the kernel directly, 
		\[
		\langle \Psi(z), \Psi(z') \rangle_{\mathcal{H}} = \kappa(z, z').
		\]  
		
		Fig. \ref{fig02} illustrates the RKHS constructed on the complex manifold. We assume that the set of vertices $\mathcal{V}$ is a collection of sampling points from a  complex manifold \( \mathcal{M} \) embedded in the complex space $\mathcal{Z}$, i.e., $\mathcal{V} \subset \mathcal{M} \subset \mathcal{Z}$. Once a kernel function $\kappa$ is defined, it induces a unique RKHS over $\mathcal{V}$, where $\kappa: \mathcal{V} \times \mathcal{V} \to \mathbb{C}$ is the corresponding reproducing kernel, and $\mathcal{H}$ is the associated RKHS. The kernel function is then extended from the discrete set of vertices $\mathcal{V}$ to the continuous complex manifold $\mathcal{M}$, i.e., $\kappa: \mathcal{M} \times \mathcal{M} \to \mathbb{C}$. Since the graph is constructed using this kernel, it induces a unique function space defined over the continuous domain of $\mathcal{M}$. Here, the graph signal $f: \mathcal{V} \to \mathbb{C}$ is viewed as $N$ samples taken from a continuous function $f: \mathcal{M} \to \mathbb{C}$, where $f \in \mathcal{H}$.
		
		\begin{figure}[htbp]
			\begin{center}
				\begin{minipage}[t]{0.85\linewidth}
					\centering
					\includegraphics[width=\linewidth]{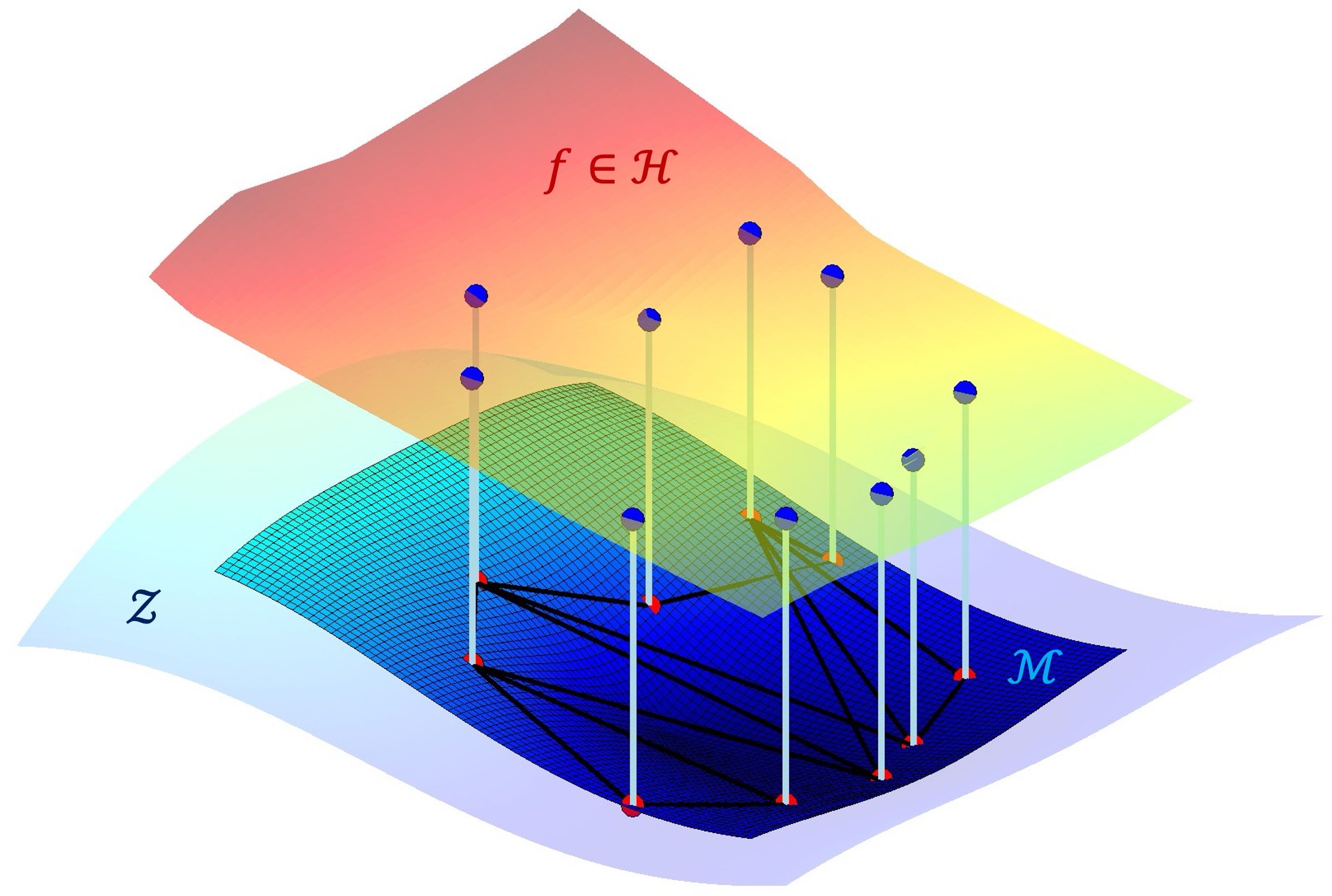}
				\end{minipage}
			\end{center}
			\caption{RKHS over complex manifold model illustration. Red points represent vertices embedded in the complex manifold $\mathcal{M}$, black lines denote edges between vertices, and light blue columns represent complex signals, extending to the blue dots, which correspond to the sampling points of the continuous function $f \in\mathcal{H}$.}
			\label{fig02}
		\end{figure}

		\subsection{Sampling and Reconstruction of Graph Signals in RKHS over Complex Manifolds}
		In practical scenarios, due to limitations in measurement resources, we adopt a sampling strategy similar to that in RKHS over differentiable manifolds. Specifically, for an $N$-dimensional graph signal $\bm{f}$, we select $|\mathcal{S}|$ samples from the sampling set $\mathcal{S} \subset \{v_0, v_1, \dots, v_{N-1}\}$, where $|\mathcal{S}| \ll N$. 
		
		In the RKHS of complex manifolds, the corresponding sampling matrix $\mathbf{D}$ and the sampled graph signal $\bm{y}_{\mathcal{S}}$ can similarly be obtained using Eqs. \eqref{D} and \eqref{yS}. Given sampled signals over the set $\mathcal{S}$, kernel-based graph signal reconstruction on a complex manifold can be formulated as an optimization problem (see Eq. \eqref{fopt}) to estimate the original signal.
		
		Utilizing the relationship $\bm{f} = \mathbf{K}\bm{\alpha}$ in Eq. \eqref{Kalpha}, where $\mathbf{K}$ is the kernel matrix and $\bm{\alpha}$ is the coefficient vector, we can estimate the kernel coefficients as follows
		\begin{equation}
			\bm{\alpha}^{\text{opt}} = \argmin_{\bm{\alpha} \in \mathbb{C}^N} \left(\frac{1}{|\mathcal{S}|}  \left| \left| \bm{y}_{\mathcal{S}} - \mathbf{D}\mathbf{K}\bm{\alpha} \right|\right|^2 + \gamma \bm{\alpha}^{\mathrm{H}} \mathbf{K} \bm{\alpha}\right).
		\end{equation}
		
		Thus, the reconstruction problem for the graph signal $\bm{f}$ on the complex manifold $ \mathcal{M}$ can be restated as
		\begin{equation}
			\bm{f}^{\text{opt}}  = \argmin_{\bm{f} \in \operatorname{span_c}\{\mathbf{K}\}} \left(\frac{1}{|\mathcal{S}|}  \left| \left| \bm{y}_{\mathcal{S}} - \mathbf{D}\bm{f} \right|\right|^2 + \gamma \bm{f}^{\mathrm{H}} \mathbf{K}^{\dagger} \bm{f}\right), \label{fopt2}
		\end{equation}
		which is a KRR estimator. 
		
		To address the optimization problem, we invoke the representer theorem \cite{RKHS}.
		
		\begin{thm} \label{thm1}
			In the complex manifold $\mathcal{M}$, let $\kappa: \mathcal{M} \times \mathcal{M} \to \mathbb{C}$ be a kernel  with the corresponding RKHS $\mathcal{H}$. Given a set of sample points $\mathcal{S} \subset  \{z_1, z_2, \ldots, z_N\} \subset \mathcal{M}$, the solution to the optimization problem is 
			\[
			f^{\text{opt}}(z) = \sum^{|\mathcal{S}|}_{m = 1} \beta_m \kappa(z, z_m),
			\]
			where \(\beta_m \in \mathbb{C}\) represents the sampled entries of \(\alpha_n\) corresponding to the indices in \(\mathcal{S}\),  which form a vector \(\bm{\beta}=\mathbf{D}\alpha\). Thus, the solution is given by
			\[
			\bm{f}^{\text{opt}} = \mathbf{K}_{\mathcal{S}} \bm{\beta},
			\]
			where $\mathbf{K}_{\mathcal{S}} = \mathbf{D}\mathbf{K}\mathbf{D}^\top$ is an $|\mathcal{S}| \times |\mathcal{S}|$ matrix, defined using the kernel matrix $\mathbf{K}$ and the sampling matrix $\mathbf{D}$.
		\end{thm}
		\begin{proof}
			See \textbf{Appendix} \ref{appendix B}.
		\end{proof}
		
		The representer theorem states that the original signal can be reconstructed from the lower-dimensional coefficient vector $\bm{\beta}$ and the corresponding columns of the kernel matrix $\mathbf{K}$ indexed by the sampling set. The optimal $\bm{\beta}$ is obtained by solving the following optimization problem
		\[
		\bm{\beta}^{\text{opt}} = \argmin_{\bm{\beta} \in \mathbb{C}^{|\mathcal{S}|}}
		\left(\frac{1}{|\mathcal{S}|}  \left| \left| \bm{y}_{\mathcal{S}} -\mathbf{K}_{\mathcal{S}}\bm{\beta} \right| \right|^2 + \gamma \bm{\beta}^{\mathrm{H}} \mathbf{K}_{\mathcal{S}}\bm{\beta}\right).
		\]
		
		If the matrix $\mathbf{K}_{\mathcal{S}}=\mathbf{DKD}^{\top}$ if of  full rank, then the KRR estimate is given by 
		\[
		\bm{f}^{\text{opt}}_{\text{KRR}} = \mathbf{KD}^{\top}\left(\mathbf{K}_{\mathcal{S}} + \gamma \mathbf{I}\right)^{-1} \bm{y}_{\mathcal{S}}.
		\]
		
		
		\section{Kernels on Complex Manifolds for Graph Signal Reconstruction}
		\label{kernel-functions}
		Kernel methods play an important role in graph signal processing, particularly in the context of reconstructing missing or incomplete graph signals \cite{Kernel}. The kernels we explore here provide powerful tools for modeling the relationships between graph vertices and the underlying data. In this section, we introduce various types of kernels, including traditional kernels, graph-based kernels, and other specialized kernels tailored to specific graph signal characteristics.
		
		
		\subsection{Traditional Kernels}
		\label{traditional-kernels}
		Traditional kernel functions, such as the Gaussian and Laplacian kernels, are widely used in machine learning and signal processing. These kernels rely on specific metrics or distances to measure the similarity between points or signals.  In the context of complex signals on complex manifolds, two commonly used metrics are the Euclidean metric and the Hermitian metric.
		
		For complex manifolds, where each pair of points is represented by local coordinates $z_n$ and $z_m$, we assume that the corresponding complex feature vectors are denoted as $\mathbf{z}_n=\left(\mathrm{z}_n^1,\mathrm{z}_n^2,\ldots,\mathrm{z}_n^D\right)^{\top}$ and $\mathbf{z}_m=\left(\mathrm{z}_m^1,\mathrm{z}_m^2,\ldots,\mathrm{z}_m^D\right)^{\top}$. The Euclidean distance is  defined as
		\begin{equation}
			\begin{aligned}
				d_E(z_n, z_m) &= \sqrt{\|\mathbf{z}_n - \mathbf{z}_m\|^2} \\
				&= \sqrt{\sum_{i=1}^{D} |\mathrm{z}_n^i - \mathrm{z}_m^i|^2}. \label{Euclidean}
			\end{aligned}
		\end{equation}
		
		The Hermitian metric, tailored for complex manifolds, accounts for the conjugate symmetry of complex numbers. For points $z_n$ and $z_m$, the distance induced by this metric reads
		\begin{equation}
			\begin{aligned}d_{H}(z_{n},z_{m})&=\sqrt{(\mathbf{z}_{n} -\mathbf{z}_{m} )^{\mathrm{H}}\mathbf{H}(\mathbf{z}_{n} -\mathbf{z}_{m} )}\\ &=\sqrt{ h_{i\bar{j} }(\mathrm{z}^{i}_{n} -\mathrm{z}^{i}_{m} )\overline{(\mathrm{z}^{j}_{n} -\mathrm{z}^{j}_{m} )}} , \end{aligned}  \label{Hermitian}
		\end{equation}
		where the metric tensor $h_{i\bar{j}}$ encapsulates the local geometric structure of the manifold. More details can be found in \textbf{Appendix} \ref{appendix A}.
		
		Different forms of $h_{i\bar{j}}$ reflect various geometric characteristics of the manifold, such as flatness, curvature, or periodicity. The choice of $h_{i\bar{j}}$ depends on the specific application scenario and the requirements of signal processing. Table~\ref{tab} summarizes several representative types of metric tensors $h_{i\bar{j}}$ and their applications.
		
		\begin{table*}[htbp]
			\centering
			\caption{Some examples of $h_{i\bar{j}}$, where $\delta_{i\bar{j}}$ denotes the Kronecker delta, i.e., $\delta_{i\bar{j}} = 1$ if $i = j$ and $0$ otherwise.}\label{tab}
			\begin{tabular}{ccc}
				\toprule
				Type & $h_{i\bar{j}}$ Expression &  Application Scenario \\
				\midrule
				Euclidean Metric & $\delta_{i\bar{j}}$    & Flat manifolds, basic signal modeling \vspace{0.15cm} \\
				Hermitian Torus Metric  & $\delta_{i\bar{j}} r^{2(i-1)}$ & Periodic signal modeling, complex tori \vspace{0.15cm}\\
				K\"ahler Metric   & $\frac{\delta_{i\bar{j}}}{\left(1+\sum\limits_{k=1}^{D} |\mathrm{z}_{n}^{k}|^2 \right)^{2}}$   & Constant-curvature complex geometry\vspace{0.15cm}\\
				Fubini-Study Metric & $\frac{\delta_{i\bar{j}}}{1+\sum\limits_{k=1}^{D} |\mathrm{z}_{n}^{k}|^2 } - \frac{\mathrm{z}_{n}^{i} \overline{\mathrm{z}}_{n}^{j}}{\left(1+\sum\limits_{k=1}^{D} |\mathrm{z}_{n}^{k}|^2 \right)^2}$  & Projective spaces, quantum state analysis\vspace{0.15cm}\\		
				Poincar\'e Metric &  $\frac{4\delta_{i\bar{j}}}{\left(1 -\sum\limits_{k=1}^{D} |\mathrm{z}_{n}^{k}|^2\right)^2}$  & Hyperbolic geometry, unit disk modeling \vspace{0.15cm}\\
				\bottomrule
			\end{tabular}
		\end{table*}
		
		These metrics or distances provide the foundation for defining kernel functions, enabling effective GSP and reconstruction on complex manifolds.

		\subsubsection{Gaussian Kernel}
		The Gaussian kernel is a classical kernel, which has been widely used in machine learning and signal processing. In complex signal processing, the Gaussian kernel can be defined using different distances. According to Eq. \eqref{Euclidean}, the traditional Euclidean Gaussian kernel (EGK) is defined as 
		\begin{equation}
			\kappa (z_{n},z_{m})=\exp \left( -\frac{d_E^2(z_n,z_m)}{2\sigma^{2} } \right), \label{LE} 
		\end{equation}
		where parameter $\sigma^2 > 0$ controls the kernel width.
		
		We can also define Gaussian kernels based on  Hermitian metrics. The Hermitian Gaussian kernel (HGK) is defined as
		\begin{equation}
			\kappa (z_{n},z_{m}) = \exp \left( -\frac{d^2_H(z_n,z_m)}{2\sigma^{2} } \right). \label{LH}
		\end{equation}

		\subsubsection{Laplacian Kernel}
		The Euclidean Laplacian kernel (ELK) is another widely used kernel, particularly effective for signals exhibiting exponential decay or sharp transitions. It is given by
		\begin{equation}
			\kappa(z_n, z_m)=\exp\left(-\frac{d_E(z_n,z_m)}{\sigma}\right),
		\end{equation}
		where $\sigma > 0$ is a scaling parameter. Unlike the Gaussian kernel, the Laplacian kernel emphasizes the similarity between nearest neighbors, making it well suited for signals with abrupt transitions or sparsely distributed data points.
		
		Similarly,  the Hermitian Laplacian kernel (HLK) is defined as 
		\begin{equation}
			\kappa (z_{n},z_{m})=\exp \left( -\frac{d_H(z_n,z_m)}{\sigma } \right) .
		\end{equation}

		\subsubsection{Polynomial Kernels}
		Polynomial kernels (PKs) capture the interaction between points nonlinearly to a certain power, which can also be naturally extended to complex manifolds equipped with Hermitian metrics.
		
		Using Eq. \eqref{Hermitian}, PKs can be defined as
		\begin{equation}
			\kappa(z_n, z_m) = \left( \langle \mathbf{z}_n, \mathbf{z}_m \rangle + c \right)^d = \left({\mathbf{z}}_n^{\mathrm{H}}\mathbf{H}\mathbf{z}_m + c \right)^d,
		\end{equation}
		where $c \geq 0$ is a constant that controls the offset, and $d$ is the degree of the polynomial.
		
		Note that when $c = 0$ and $d = 1$, PKs reduce to the linear kernel.
		\begin{example}
			As an example, consider that the vertices embedded in the complex manifold are assigned features $\mathbf{z}_n$, defined as 
			$$\mathrm{z}_n^i = \mathrm{x}_n + \mathrm{i} \cdot \cos\left(2 \pi \frac{i}{D} \mathrm{x}_n\right),$$
			where $\mathrm{x}_n = -1 + \frac{2(n-1)}{N-1}, n = 1,2, \dots, N$ and $i = 1, 2,\dots, D$. Based on these features, the graph signal is generated using Eq. \eqref{Kalpha}. The noisy graph signals are reconstructed using the five aforementioned kernels and KRR, with $h_{i\bar{j}}$ defined under the Kähler metric and Poincaré metric. The results are summarized in Fig. \ref{fig03}.
		\end{example}
		
		\begin{figure}[htbp]
			\begin{center}
				\begin{minipage}[t]{0.49\linewidth}
					\centering
					\includegraphics[width=\linewidth]{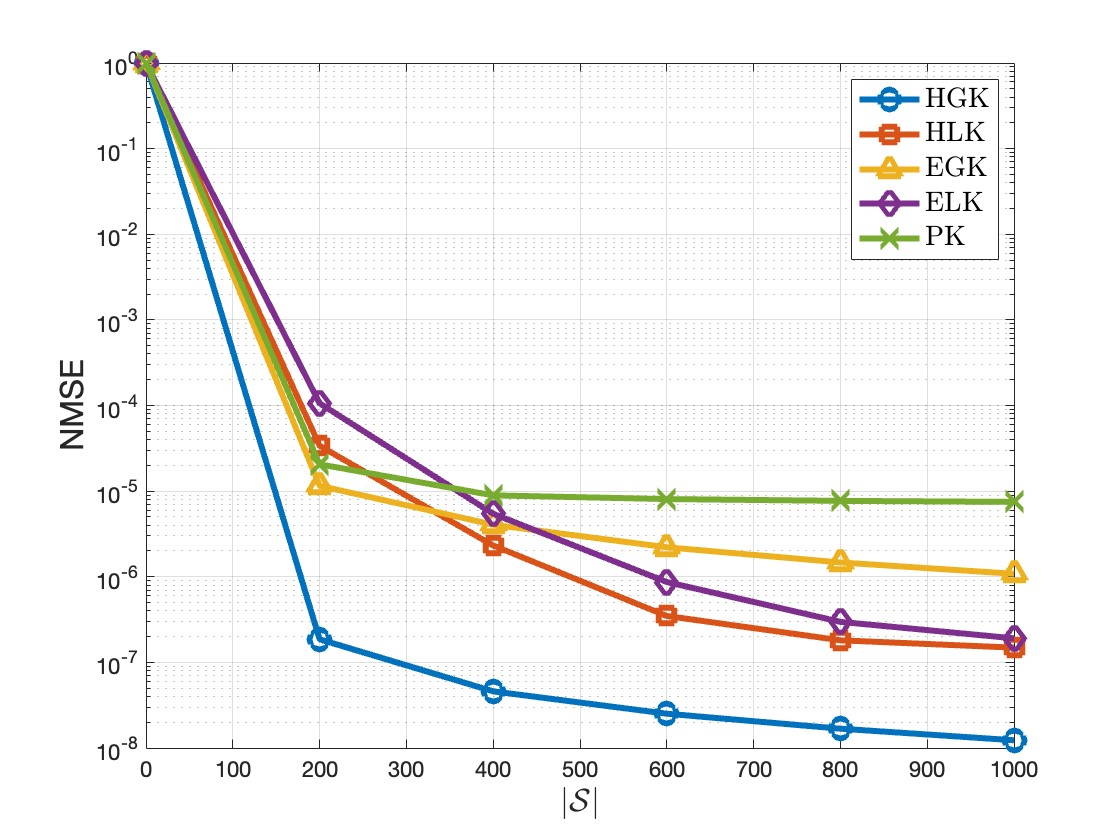}
					\parbox{4cm}{\tiny (a) Kähler metric on complex manifold, with Gaussian noise $\mathcal{N}(0, 0.1^2)$.}
				\end{minipage}
				\begin{minipage}[t]{0.49\linewidth}
					\centering
					\includegraphics[width=\linewidth]{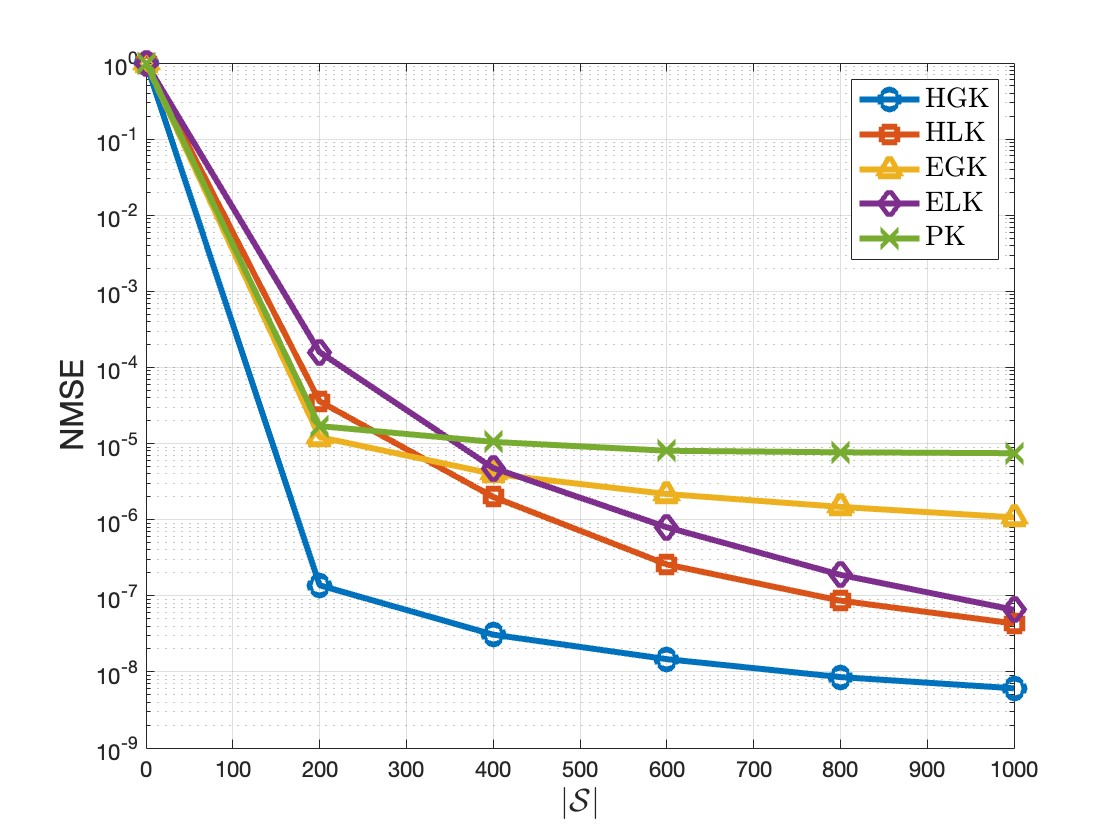}
					\parbox{4cm}{\tiny (b) Kähler metric on complex manifold, with Gaussian noise $\mathcal{N}(0, 0.05^2)$.}
				\end{minipage}
				\begin{minipage}[t]{0.49\linewidth}
					\centering
					\includegraphics[width=\linewidth]{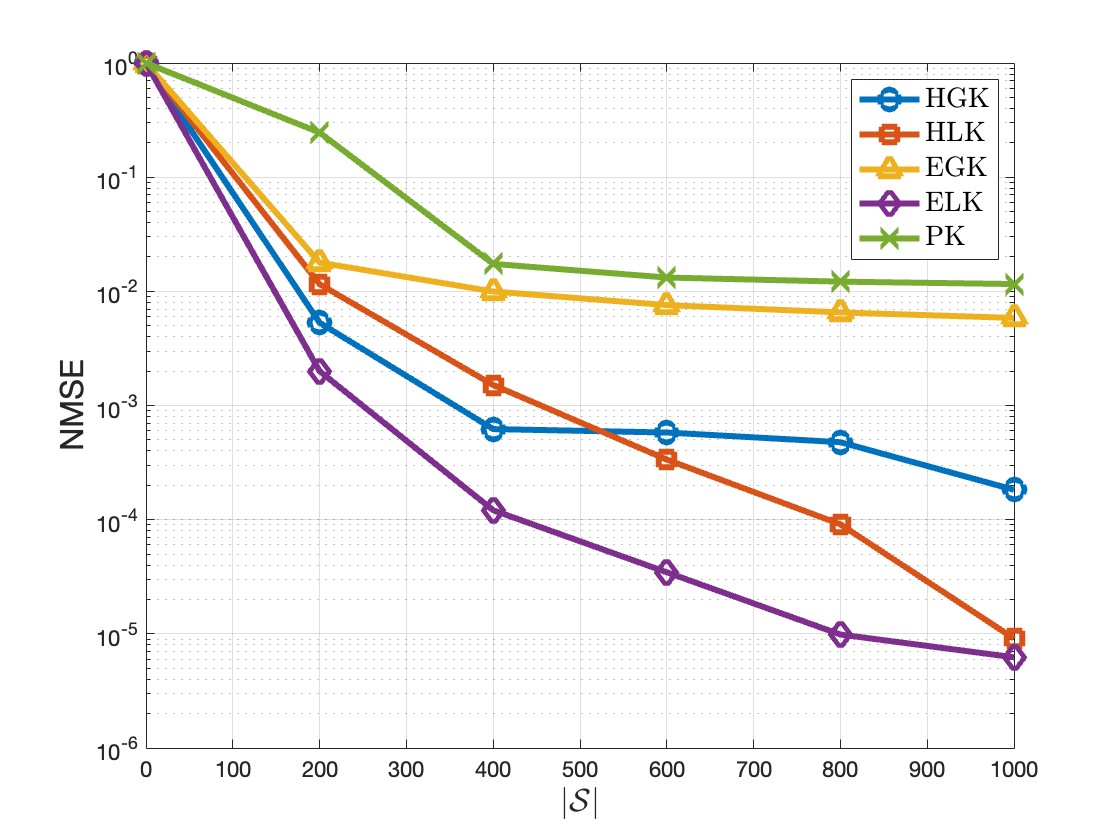}
					\parbox{4cm}{\tiny (c) Poincaré metric on complex manifold, with Gaussian noise $\mathcal{N}(0, 0.1^2)$.}
				\end{minipage}
				\begin{minipage}[t]{0.49\linewidth}
					\centering
					\includegraphics[width=\linewidth]{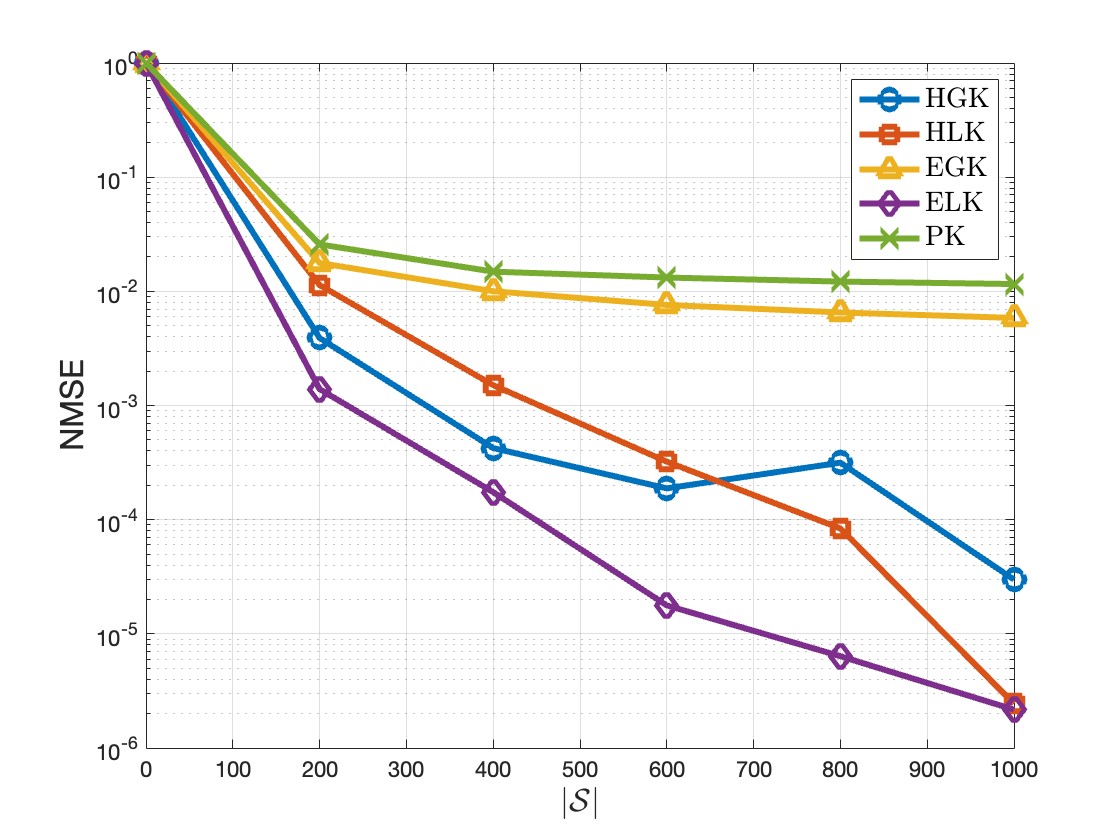}
					\parbox{4cm}{\tiny (d) Poincaré metric on complex manifold, with Gaussian noise $\mathcal{N}(0, 0.05^2)$.}
				\end{minipage}
			\end{center}
			\caption{The normalized mean squared error (NMSE) between the reconstructed and original signals is evaluated across five kernels with an increasing number of sampling nodes. The kernel parameters are set to $\sigma = 0.5$, $c = 10$, and $d = 8$, with $N = 1000$, $D = 2$, and $\gamma =0.01$. To account for the randomness in sampling node selection, each experiment is repeated 20 times, and the average results are reported.}
			\label{fig03}
		\end{figure}
		
		Embedding vertices into specific complex manifolds demonstrates that kernels based on Hermitian metrics can outperform those based on the Euclidean metric. Thus, selecting an appropriate manifold or metric is crucial in kernel-based graph signal reconstruction. 
		
		
		\subsection{Graph-based Kernels}
		\label{graph-kernels}
		Embedding graph signals into a complex manifold and applying traditional kernels for graph signal reconstruction typically relies on the similarity of features between points on the manifold. However, for a given graph, where explicit relationships exist between vertices, graph kernels are commonly used to leverage the graph topological structure, facilitating smooth estimations \cite{GFTlaplace}.

		\subsubsection{Graph Laplacian Kernel}
		On complex manifolds, the similarity between adjacent vertices---or vertices that are close in terms of geodesic distance---can be used to construct graph kernels that capture the topological structure of the graph \cite{GFTlaplace,GFTadjacency2,Goverview}. The graph Laplacian kernel (GLK) fis defined by applying a specific function to the Laplacian matrix $\mathbf{L}$. It can be indirectly constructed using the traditional kernel method described in the previous section. 
		
		We first construct the adjacency and degree matrices of a given graph:
		\begin{equation}
			\begin{aligned}
				w_{nm}& = \kappa(z_n, z_m),\\
				d^{\mathrm{deg}}_n &= \sum_m w_{nm} = \sum_m \kappa(z_n, z_m), \label{wnm}
			\end{aligned}
		\end{equation}
		where $\kappa(z_n, z_m)$ is a Gaussian kernel, either  \eqref{LE} or \eqref{LH}. Then, we obtain the graph Laplacian matrix 
		\begin{equation}
			{\mathbf{L}_{nm}} = \begin{cases}
				d^{\mathrm{deg}}_n, & \text{if } n = m, \\
				-w_{nm}, & \text{else. }
			\end{cases}
		\end{equation}
		
		Before defining the GLK, we  introduce the graph fractional Fourier transform (GFRFT) \cite{GFRFT} of a graph signal $\bm{f} \in \mathcal{H}$, \begin{equation}
			\hat{\bm{f}} = (\mathbf{U}^{\mathrm{H}})^a \bm{f}, \quad a \in \mathbb{R}, \label{GFRFT}
		\end{equation}
		where $\mathbf{U}$ represents  the matrix of eigenvectors derived from the eigenvalue decomposition of the Laplacian matrix $\mathbf{L}$\cite{GFRFTspectral}, namely
		\begin{equation}
			\mathbf{L} = \mathbf{U \Lambda U}^{\mathrm{H}},
			\label{La}
		\end{equation}
		where the matrix $\mathbf{U} = [{\bm{u}}_1,\bm{u}_2, \dots, {\bm{u}}_N]$ is composed of the eigenvectors, and $\mathbf{\Lambda} = \text{diag}(\lambda_1,\lambda_2, \dots, \lambda_N)$, with eigenvalues $\lambda_1, \lambda_2,\dots, \lambda_N$ arranged in descending order.
		
		Through GFRFT, we are able to not only analyze the frequency-domain characteristics of complex signals but also transform these signals in a flexible manner, unveiling the deeper structural relationships of the graph signal on the complex manifold. The graph fractional Laplacian matrix demonstrates superior performance in capturing prior information about complex signals \cite{GFRFT,HGFRFT,GFRFTspectral,GLCT,newGLCT}.
		
		Building on the eigenvalue decomposition of the Laplacian matrix in Eq. \eqref{La}, the kernel matrix derived from the GLK is given by 
		\begin{equation}
			\begin{aligned}
				\mathbf{K}&=r^{\dag }\left( \mathbf{L}^{a}\right)  \\
				&=r^{\dag }\left( \mathbf{U}^{a}\mathbf{\Lambda}^{a} ( \mathbf{U}^{\mathrm{H}})^{a}  \right)  \\ & =\mathbf{U}^{a}r^{\dag }\left( \mathbf{\Lambda}^{a} \right)  ( \mathbf{U}^{\mathrm{H}})^{a},\label{GLK}
			\end{aligned}
		\end{equation}  
		where $r^\dagger(\mathbf{\Lambda}^{a})$ applies a user-defined non-negative mapping $ r(\lambda^{a})$ to the eigenvalues $\lambda^a$ of $\mathbf{L}^{a}$. Typical choices for $ r(\lambda^{a})$ include \cite{GFTlaplace,Diffusionkernel,regularizationkernel}
		\begin{itemize}
			\item{Diffusion kernel: $r(\lambda^{a}) = \exp(-\sigma^2 \lambda^{a} / 2)$,}
			\item{Random walk kernel: $ r(\lambda^a) = (b - \lambda^{a})^{-p} $, where $ b \geq 2 $,}
			\item{Laplacian regularization: $ r(\lambda^a) = 1 + \sigma^2 \lambda^{a} $.}
		\end{itemize}
		These kernel functions leverage the spectral properties of the graph fractional order to capture its structural information, enabling smooth estimations of graph signals, particularly in complex manifolds.
		
		When the GLK is incorporated into the regularizer of an optimization problem, we obtain
		\begin{equation}
			\begin{aligned}
				\bm{f}^{\mathrm{H}} \mathbf{K}^{\dagger} \bm{f} &= \bm{f}^{\mathrm{H}} r\left( \mathbf{L}^{a} \right) \bm{f}\\
				&= \bm{f}^{\mathrm{H}} \mathbf{U}^{a} r\left( \mathbf{\Lambda}^{a} \right) (\mathbf{U}^{\mathrm{H}})^{a} \bm{f} \\
				&= \hat{\bm{f}}^{\mathrm{H}} r\left( \mathbf{\Lambda}^{a} \right) \hat{\bm{f}}, 
			\end{aligned}
		\end{equation}  
		which shows that this term can be expressed in terms of the transformed signal in the GFRFT domain. When this regularizer is specialized to the form $\bm{f}^{\mathrm{H}} \mathbf{L}^{a} \bm{f}$, we obtain what is referred to as the fractional-order Laplacian regularizer.
		\begin{proposition}\label{pro:smo}
			For the smoothness of a graph embedded in a complex manifold, we have
			\begin{equation}\label{eq:flf}
				\bm{f}^{\mathrm{H}} \mathbf{L}^{a} \bm{f} = \frac{1}{2} \sum^{N}_{n=1} \sum^{N}_{m=1} w^{a}_{nm} \left( f\left( z_{n} \right) - f\left( z_{m} \right) \right)^2,
			\end{equation}
			where $w^{a}_{nm}$ is the kernel  defined in Eq. \eqref{wnm}.
		\end{proposition}
		\begin{proof}
			See \textbf{Appendix} \ref{appendix C}.
		\end{proof}
		
		The smoothness measure depends on both the spectral components of the signal (i.e., eigenvectors and eigenvalues) and the fractional power $a$. Incorporating $a$ amplifies the influence of eigenvalues, especially in the high-frequency domain, where it suppresses high-frequency components to enhance smoothness \cite{GFRFTspectral}. For the Laplacian matrix with fractional power $a$, the kernel function $r(\lambda^a)$ imposes a stronger smoothness penalty on low-frequency components, while penalizing high-frequency components less, resulting in a smoother signal on the complex manifold.
		
		Furthermore, using the inverse transform of Eq. \eqref{GFRFT}, we have
		\begin{equation}
			\bm{f} = \mathbf{U}^{a} \hat{\bm{f}} = \langle \hat{\bm{f}}, \bm{u}^{a} \rangle = \sum^{N}_{n=1} \hat{f}_{n} \bm{u}^{a}_{n},  
		\end{equation}  
		indicating that the graph signal $\bm{f}$ in the RKHS can be decomposed into components embedded within the eigenvectors $\mathbf{U}^{a}$ of the Laplacian matrix.

		\subsubsection{Graph Bandlimited Kernel}
		Graph bandlimited kernels (GBKs) effectively characterize the spectral characteristics of signals. 
		
		In GSP, graph bandlimited signals are defined via localization operators. For a vertex subset $\mathcal{S} \subseteq \mathcal{V}$ and a spectral subset $\mathcal{F} \subseteq \hat{\mathcal{G}}$, where $\hat{\mathcal{G}} = \{1,2, \ldots, N\}$ represents the set of all spectral indices, the graph localization operators are defined as 
		\begin{equation}
			\mathbf{P} = \text{diag}(p_1, p_2, \dots, p_N), \ \text{and} \ \mathbf{B}^{a} = \mathbf{U}^{a} \mathbf{\Sigma} (\mathbf{U}^{\mathrm{H}})^{a}, \label{PBa}
		\end{equation}
		where $p_i = 1$ if $i \in \mathcal{S}$, and otherwise $p_i = 0$. Furthermore, the matrix $ \mathbf{\Sigma} = \text{diag}(\sigma_1, \sigma_2, \dots, \sigma_N)$ is defined such that $\sigma_i = 1$ if $i \in \mathcal{F}$, and otherwise $\sigma_i = 0$.
		
		\begin{remark}
			It should be noted that the vertex localization operator $\mathbf{P}=\mathbf{D}^{\top}\mathbf{D}$, where $\mathbf{D}$ is the sampling operator given by  Eq. \eqref{D}. The definition of the fractional spectral localization operator $\mathbf{B}^{a}$ is analogous to the graph Laplacian kernel. Specifically, it coincides with the Laplacian kernel when $r^\dagger(\mathbf{\Lambda}^a) = \mathbf{\Sigma}$.
		\end{remark}
		
		The operators $\mathbf{P}$ and $\mathbf{B}^{a}$ are Hermitian positive semi-definite, with spectral norms equal to $1$.
		
		\begin{lem}
			\label{lem1}
			A vector $\bm{f}$ is fully localized in the vertex set $\mathcal{S}$ and the spectral set $\mathcal{F}$ if and only if 
			\[
			\lambda_{\max} \left( \mathbf{B}^{a} \mathbf{P} \mathbf{B}^{a} \right) = 1
			\]
			satisfies.
			In this case, $\bm{f}$ is an eigenvector corresponding to the unit eigenvalue and satisfies $\mathbf{P} \bm{f} = \bm{f}$, and $\mathbf{B}^{a} \bm{f} = \bm{f}$.
		\end{lem}
		\begin{proof}
			\textbf{Appendix} \ref{appendix D}.
		\end{proof}
		
		Lemma \ref{lem1} implies that, if a graph signal $\bm{f}$ is bandlimited, it can be written  as
		\begin{equation}
			\bm{f} = \mathbf{B}^{a}\bm{f} = \mathbf{U}^a \mathbf{\Sigma} (\mathbf{U}^{\mathrm{H}})^a \bm{f} = \mathbf{U}^a \mathbf{\Sigma} \hat{\bm{f}} = \mathbf{U}^a_{\mathcal{F}} \hat{\bm{f}}_{\mathcal{F}}, \label{Baf}
		\end{equation}
		where $\mathbf{U}^a_{\mathcal{F}}$ consists of the columns of $\mathbf{U}^a$ corresponding to the indices in $\mathcal{F}$, and $\hat{\bm{f}}_\mathcal{F}$ is the vector of coefficients $\hat{\bm{f}}_i$ for $i \in \mathcal{F}$.
		
		The optimization problem \eqref{fopt2} can then be rewritten in the transform domain. Firstly, for the regularization term, the construction of the GBK is based on the GLK in Eq. \eqref{GLK}, defined by
		\[
		r_{\mu}(\lambda^a_i) =
		\begin{cases}
			\frac{1}{\mu}, & \text{if } i \in\mathcal{F}, \\
			\mu, & \text{else}.
		\end{cases}
		\]
		For $\hat{\bm{f}}_i$, since $i \notin \mathcal{F}$, a penalty $\mu$ is imposed to make it less influential, thus promoting bandlimited estimation. Set $1/\mu$ for $i \in \mathcal{F}$ to ensure that $\mathbf{K}_\mu$ is non-singular \cite{Kernel}.
		
		Since both $\mathbf{P}$ and $\mathbf{B}^a$ are self-adjoint, and they satisfy $\mathbf{P}\bm{f} = \bm{f} = \mathbf{B}^a \bm{f}$, the optimization problem can be rewritten as
		\begin{equation}
			\bm{f}^{\text{opt}} = \mathbf{U}^a_{\mathcal{F}} \argmin_{\hat{\bm{f}}_{\mathcal{F}} \in \mathbb{C}^{|\mathcal{F}|}} \left( \| \bm{y}_{\mathcal{S}} - \mathbf{D} \mathbf{U}^a_{\mathcal{F}} \hat{\bm{f}}_\mathcal{F} \|^2+\gamma \hat{\bm{f}}^{\mathrm{H}}_{\mathcal{F}}r_{\mu}(\mathbf{\Lambda}^{a})\hat{\bm{f}}_{\mathcal{F}}\right), \label{bandlimitedOPT}
		\end{equation}
		where $\mathbf{\Lambda}^a = \text{diag}(\lambda_1^a, \lambda_2^a, \dots, \lambda_N^a)$. Stationary condition of the optimization problem Eq. \eqref{bandlimitedOPT} yields the optimal graph bandlimited signal, which is 
		\[
		\begin{aligned}\bm{f}^{\text{opt} }_{\text{rr} } &=\mathbf{U}^{a}_{\mathcal{F} } \hat{\bm{f} }_{\mathcal{F} } \\ &=\mathbf{U}^{a}_{\mathcal{F} } \left( (\mathbf{U}^{a}_{\mathcal{F} } )^{\mathrm{H} }\mathbf{P} \mathbf{U}^{a}_{\mathcal{F} } +\gamma r_{\mu }(\mathbf{\Lambda }^{a} )\right)^{-1}  (\mathbf{U}^{a}_{\mathcal{F} } )^{\mathrm{H} }\mathbf{D}^{\top} \bm{y}_{\mathcal{S} }. \end{aligned} 
		\]
		

		\section{Multi-Kernel Learning on Complex Manifolds}
		\label{MKL}
		Signal distributions on complex manifolds often exhibit both local and global characteristics, making it challenging to fully capture their diverse features using a single kernel. Moreover, in practical applications, uncertainties in kernel types, parameter selection, and  bandwidth of bandlimited graph signals add further complexity. To address these issues, MKL provides a versatile framework by combining multiple kernels to adapt to intricate geometrical structures \cite{MKL}.
		
		This section introduces a novel graph signal reconstruction method based on MKL in complex manifolds. Leveraging the MKL framework, the proposed approach integrates inner and outer optimizations to preserve sparse regularization and multi-kernel weighting strategies while reducing  optimization dimensions. This design minimizes computational complexity, ensuring robust and efficient reconstruction.
		
		In a complex manifold $\mathcal{M}$, we define a family of kernel functions $\{\kappa_{\ell}(z, z')\}_{\ell=1}^L$, where each kernel $\kappa_{\ell}$ corresponds to a RKHS $\mathcal{H}_\ell$. These kernels may capture various signal characteristics, such as local, global, and spectral properties. For example, Gaussian kernels are well-suited for modeling local features, while bandlimited kernels excel in representing graph spectral domain information.
		
		A graph signal $f: \mathcal{V} \to \mathbb{C}$ is represented as a weighted combination of these kernels, namely, 
		\begin{equation}
			f(z) = \sum_{\ell=1}^L f_\ell(z), \quad f_\ell \in \mathcal{H}_\ell,
		\end{equation}
		where, by using the Theorem \ref{thm1}, each $f_\ell$ can be formulated over the observed subset $\mathcal{S} \subseteq \mathcal{V}$ as
		\begin{equation}
			f_\ell(z) = \sum_{s=1}^{|\mathcal{S}|} (\beta_{\ell})_{s} \kappa_\ell(z, z_s)
		\end{equation} 
		with $\beta_{\ell}$ the optimization parameters and $\{z_s\}_{s=1}^{|\mathcal{S}|}\subseteq \mathcal{S}$ the observable graph vertices embedded in the complex manifold. The corresponding observed values are given by $\bm{y}_{\mathcal{S}} = (y_1, y_2, \ldots, y_{|\mathcal{S}|})^\top$, while $\bm{f}_{\mathcal{S}} = (f(z_1), f(z_2), \ldots, f(z_{|\mathcal{S}|}))^\top$ denotes the ground-truth values at these vertices.
		
		The goal is to reconstruct the complete signal $f$ in all vertices. This  is formulated as an optimization problem
		\[
		\argmin_{\{f_\ell\}_{\ell=1}^L}\left( \frac{1}{|\mathcal{S}|} \sum_{s=1}^{|\mathcal{S}|} \left( y_s - \sum_{\ell=1}^{L} f_{\ell}(z_s) \right)^2 + \gamma \sum_{\ell=1}^{L} \|f_\ell\|^{2}_{\mathcal{H}_\ell}\right).
		\]
		
		By substituting $f_\ell$ with the kernel expansions inside, the problem becomes to find  kernel weights $\{\bm{\beta}_\ell\}_{\ell=1}^L$ that minimize the objective function
		\begin{equation}
			\argmin_{\left\{ \bm{\beta}_{\ell } \right\}^{L}_{\ell =1} } 
			\frac{1}{|\mathcal{S} |} \left\| \bm{y}_{\mathcal{S} } -\sum^{L}_{\ell =1} (\mathbf{K}_{\ell})_{\mathcal{S}} \bm{\beta}_{\ell } \right\|^{2} +\gamma \sum^{L}_{\ell =1} \bm{\beta}^{\mathrm{H}}_{\ell } (\mathbf{K}_{\ell})_{\mathcal{S}}\bm{\beta}_{\ell },\label{eqMKL}
		\end{equation}
		where $(\mathbf{K}_{\ell})_{\mathcal{S}}=\mathbf{D}\mathbf{K}_{\ell } \mathbf{D}^{\top }$. 
		
		Due to the large number of variables in Eq. \eqref{eqMKL}, totaling $ |\mathcal{S}|\times L$, the computational complexity is significant. To mitigate this, we adopt a nested optimization framework that reduces the total number of variables while preserving the original sparse regularization. 
		
		In the optimization, let the kernel weight vector be $\bm{\omega}= \left(\omega_{1} ,\omega_2,\ldots ,w_{L}\right)^{\top}$, and select the appropriate kernel combination from the predefined kernel set $\{\mathbf{K}_\ell\}_{\ell=1}^L$ after optimization. The signal representation vector $\bm{\beta}$ is optimized over the combined kernel  $\mathbf{K}(\bm{\omega}) = \sum_{\ell=1}^L \omega_\ell \mathbf{K}_\ell$ to estimate the signal directly. For notational simplicity, we denote $\mathbf{K}_{\mathcal{S}}(\bm{\omega}) = \mathbf{D}\mathbf{K}(\bm{\omega})\mathbf{D}^{\top}$. To ensure robust kernel-based reconstruction with enhanced sparsity control, we incorporate both an explicit regularization term and a geometric constraint into the optimization problem. The reformulated objective function is given by
		\begin{equation}
			\begin{aligned}
				\left( \bm{\omega}, \bm{\beta} \right) &= \argmin_{\bm{\omega}, \bm{\beta}} \Big( \frac{1}{|\mathcal{S}|} \left\| \bm{y}_{\mathcal{S}} - \mathbf{K}_{\mathcal{S}}(\bm{\omega}) \bm{\beta} \right\|^2 \\
				&\quad \quad \quad \quad \quad + \gamma \bm{\beta}^{\mathrm{H}} \mathbf{K}_{\mathcal{S}}(\bm{\omega}) \bm{\beta} 
				+ \nu  \|\bm{\omega}\|_1\Big), \label{eqMKL2}
			\end{aligned}
		\end{equation}
		subject to the constraint
		\[
		\bm{\omega} \in \{\bm{\omega} \mid \bm{\omega} \geq 0, \|\bm{\omega} - \bm{\omega}_0\|_1 \leq R\},
		\]
		where $\nu  > 0$ controls the trade-off between sparsity and reconstruction accuracy, while the $\ell_1$-ball constraint ensures that the kernel weights $\bm{\omega}$ remain within a feasible search space centered at $\bm{\omega}_0$ with radius $R > 0$.
		
		Consequently, the MKL-based signal estimation on the graph is given by
		\[
		\bm{f}^{\text{opt}}=\mathbf{K}_{\mathcal{S}}(\bm{\omega})\bm{\beta}=\mathbf{D}\mathbf{K}(\bm{\omega})\mathbf{D}^{\top} \bm{\beta}.
		\]
		
		The optimization for $\bm{\omega}$ has a dimension of $L$, while $\bm{\beta}$ is of dimension $|\mathcal{S}|$. The degree of freedom is reduced to $L + |\mathcal{S}|$. However, \eqref{eqMKL2} introduces sparsity through regularization, partially retaining the characteristics of \eqref{eqMKL}. The proposed algorithm is summarized in Algorithm \ref{alg1}.
		
		\begin{algorithm}[htbp]
			\caption{\footnotesize{Sparse Kernel-Weighted Interpolation with $\ell_1$-Ball Constraint}}\label{alg1}
			\footnotesize
			\begin{algorithmic}
				\REQUIRE Initial kernel weights $\bm{\omega}^{(0)}$, kernel matrices $\{\mathbf{K}_\ell\}_{\ell=1}^L$, regularization parameters $\gamma, \nu $, step size $\eta$, tolerance $\epsilon$, radius $R > 0$.
				\STATE  \textbf{Initialization}:
				\STATE  \quad $\bm{\beta}^{(0)} = (\mathbf{K}(\bm{\omega}^{(0)}) + \gamma |\mathcal{S}| \mathbf{I})^{-1} \bm{y}_{\mathcal{S}}$
				\STATE  \quad $k = 0$
				\REPEAT
				\STATE  \textbf{Step 1: Update kernel weights $\bm{\omega}$}:
				\[
				\bm{\omega}^{(k+1)} = ( \bm{\omega}, \bm{\beta}^{(k)}) \quad \text{(cf. Eq.  \eqref{eqMKL2})}
				\]
				\vspace{-0.3cm}
				\STATE \quad subject to $\bm{\omega}^{(k+1)}  \geq 0$ and $\|\bm{\omega}^{(k+1)}  - \bm{\omega}^{(0)} \|_1 \leq R$.
				\STATE \quad Solve using a sparse regularization technique while enforcing the $\ell_1$-ball constraint.
				\STATE \textbf{Step 2: Update signal representation $\bm{\beta}$}:
				\[
				\bm{\beta}^{(k+1)} = \eta \bm{\beta}^{(k)} + (1 - \eta) \left(\mathbf{K}(\bm{\omega}^{(k+1)}) + \gamma |\mathcal{S}| \mathbf{I}\right)^{-1} \bm{y}_\mathcal{S}
				\]
				\vspace{-0.2cm}
				\STATE \textbf{Step 3: Check convergence condition}:
				\STATE \quad If $\|\bm{\beta}^{(k+1)} - \bm{\beta}^{(k)}\| < \epsilon$, terminate iteration.
				\STATE  $k \gets k + 1$
				\UNTIL{Convergence condition is satisfied}
				\ENSURE Optimized kernel weights $\bm{\omega}$ and signal representation $\bm{\beta}$.
			\end{algorithmic}
		\end{algorithm}

		\begin{remark}  
			Although the objective function is convex with respect to each individual function $f_\ell$, it is not jointly convex in all $\{f_\ell\}_{\ell=1}^L$. Consequently, the algorithm may converge to a local optimum \cite{Kernel,convex}. The choice of $\gamma$ and $\nu$ significantly impacts the results and should be determined via cross-validation or task-specific methodologies.  
		\end{remark}

		\section{Numerical Experiments }
		In this section, we compare the proposed methods with competing alternatives through numerical simulations on both synthetic and real-world datasets. Performance metrics are averaged across multiple realizations of the signal $\bm{f}$, noise $\bm{n}$ (applicable only to synthetic-data experiments), and sampling sets $\mathcal{S}$. The sampling sets are uniformly drawn at random without replacement from $\{1,2, \dots, N\}$, ensuring comprehensive evaluation under varying conditions.
		

		\subsection{Simulation on Synthetic Signals} \label{sec6.1}
		Here we evaluate the performance of the proposed RKHS methods on complex manifolds for graph signal reconstruction using synthetic data. The experiments are divided into two parts: one focuses on the effects of different kernel functions on complex manifolds, and the other compares estimation methods of bandlimited and non-bandlimited signals. All experiments are repeated under various settings, and reconstruction performance is measured using the NMSE between the original and reconstructed signals.
		
		\subsubsection{Bandlimited Signal Estimation Methods}
		Bandlimited graph signals are generated on graphs embedded in complex manifolds. Specifically, we consider two graphs, each with $N = 200$ nodes, generated via the Swiss roll and community models \cite{GSPBOX} to reflect manifold structures. The graph signals are constructed as linear combinations of eigenvectors of the fractional Laplacian matrix on the complex manifold, with a bandwidth $|\mathcal{F}| = 40$, as described in Eq. \eqref{Baf}. By choosing a fractional parameter $a = 0.9$ and applying a bandlimited filter, the bandlimited graph signal is defined as 
		\[
		\bm{f} = \mathbf{B}^{a}\left( 10\mathbf{U}_1 + 5\mathbf{U}_2 + 20\mathbf{U}_3\right) ,
		\]
		where $\mathbf{U}_i$ represents the $i$-th column of the eigenvector matrix $\mathbf{U}$, and $\mathbf{B}^{a}$ denotes the bandlimited filter (operator) defined in \eqref{PBa}. The filter is designed to retain only the first $|\mathcal{F}|$ frequency components, suppressing high-frequency terms to achieve the bandlimited signal.
		
		To evaluate signal reconstruction performance, we randomly sample the bandlimited signal, and add Gaussian noise $\mathcal{N}(0, 0.01^2)$. Reconstruction is conducted using the GBK method, with different kernels corresponding to varying spectral bandwidths. We compare reconstructions  using kernels with bandwidths $|\mathcal{F}| = 30, 40, 50, N$, as well as GLK without bandwidth and MKL for all GBK.
		
		For varying sampling sizes and graph fractional parameters, NMSE is computed. Sampling sizes and fractional parameters are varied, and for each case, the NMSE is averaged over 20 independent trials. Figs. \ref{fig04} (a) and (b) show the NMSE with increasing sampling size for different underlying graphs, indicating that the MKL method performs best as the sample size increases, followed by GBK with $|\mathcal{F}| = 40$. Figs. \ref{fig04} (c) and (d) show the  NMSE with different fractional parameters, revealing that the reconstruction performance is optimal at a specific fractional order ($a = 0.6$ or $0.7$).
		
		\begin{figure}[htbp]
			\begin{center}
				\begin{minipage}[t]{0.49\linewidth}
					\centering
					\includegraphics[width=\linewidth]{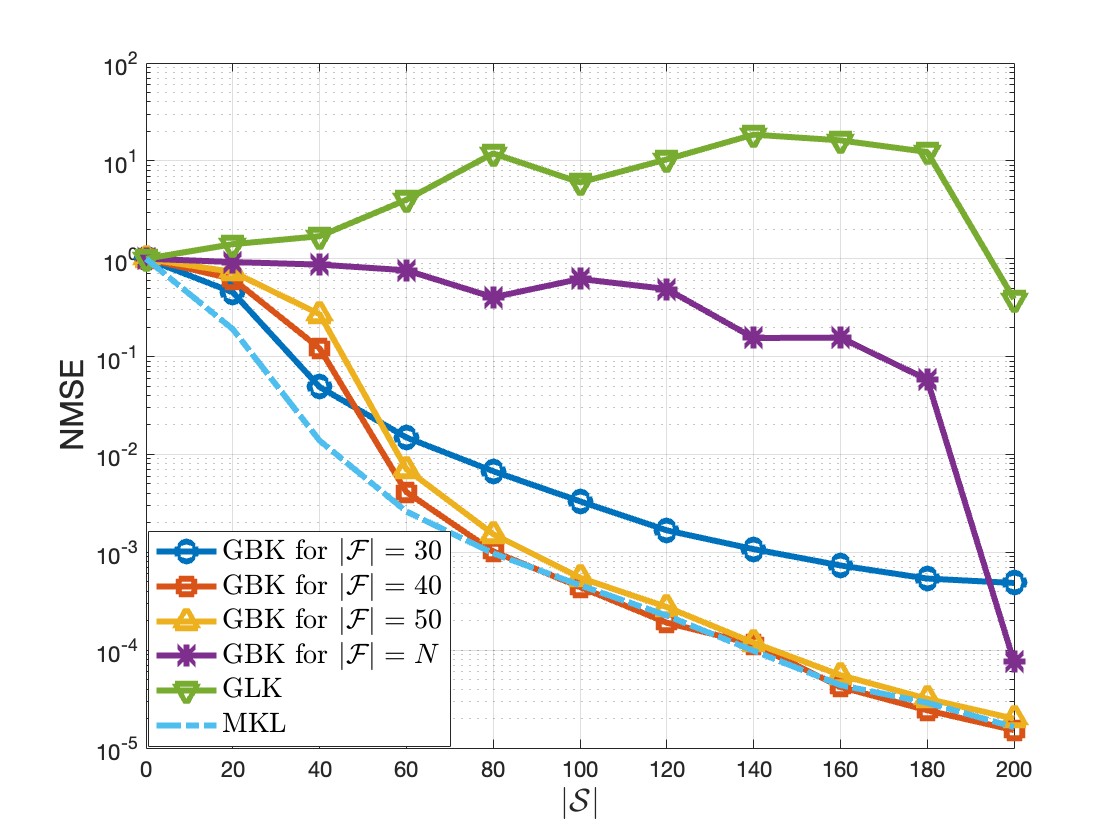}
					\parbox{3.8cm}{\tiny (a) NMSE variation with $|\mathcal{S}|$ on Swiss roll graph.}
				\end{minipage}
				\begin{minipage}[t]{0.49\linewidth}
					\centering
					\includegraphics[width=\linewidth]{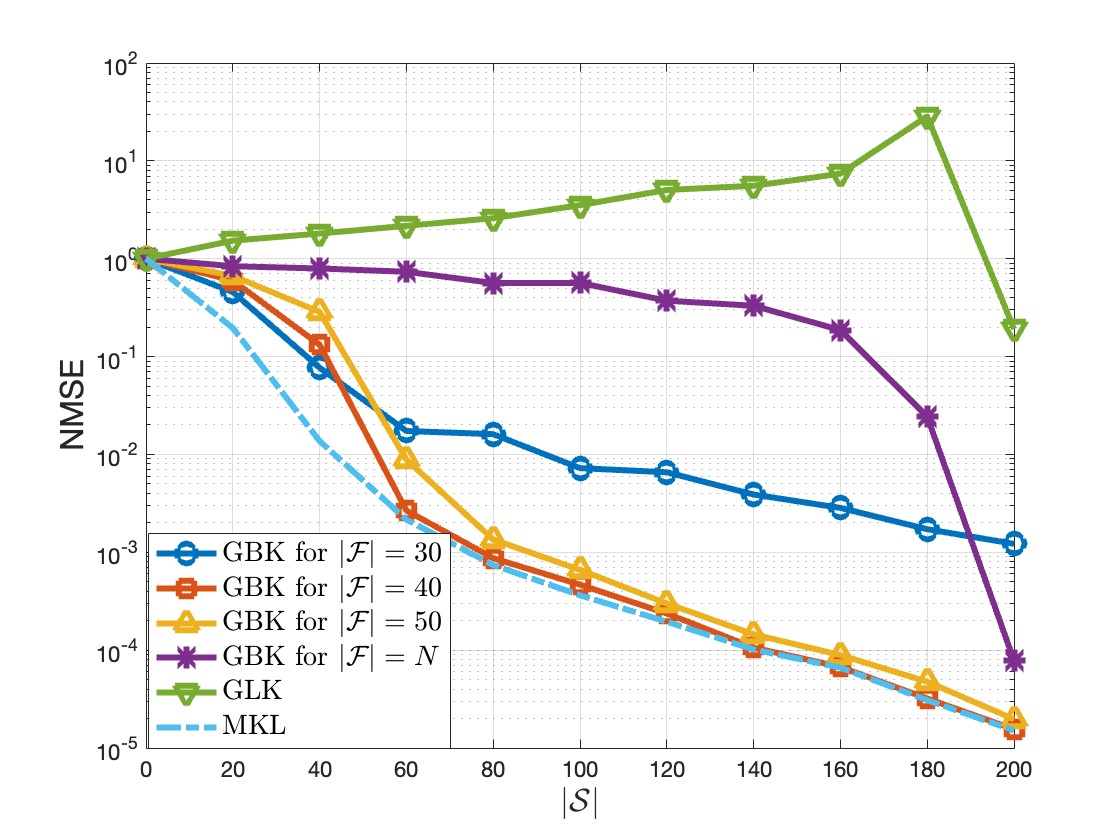}
					\parbox{3.8cm}{\tiny (b) NMSE variation with $|\mathcal{S}|$ on community graph.}
				\end{minipage}
				\begin{minipage}[t]{0.49\linewidth}
					\centering
					\includegraphics[width=\linewidth]{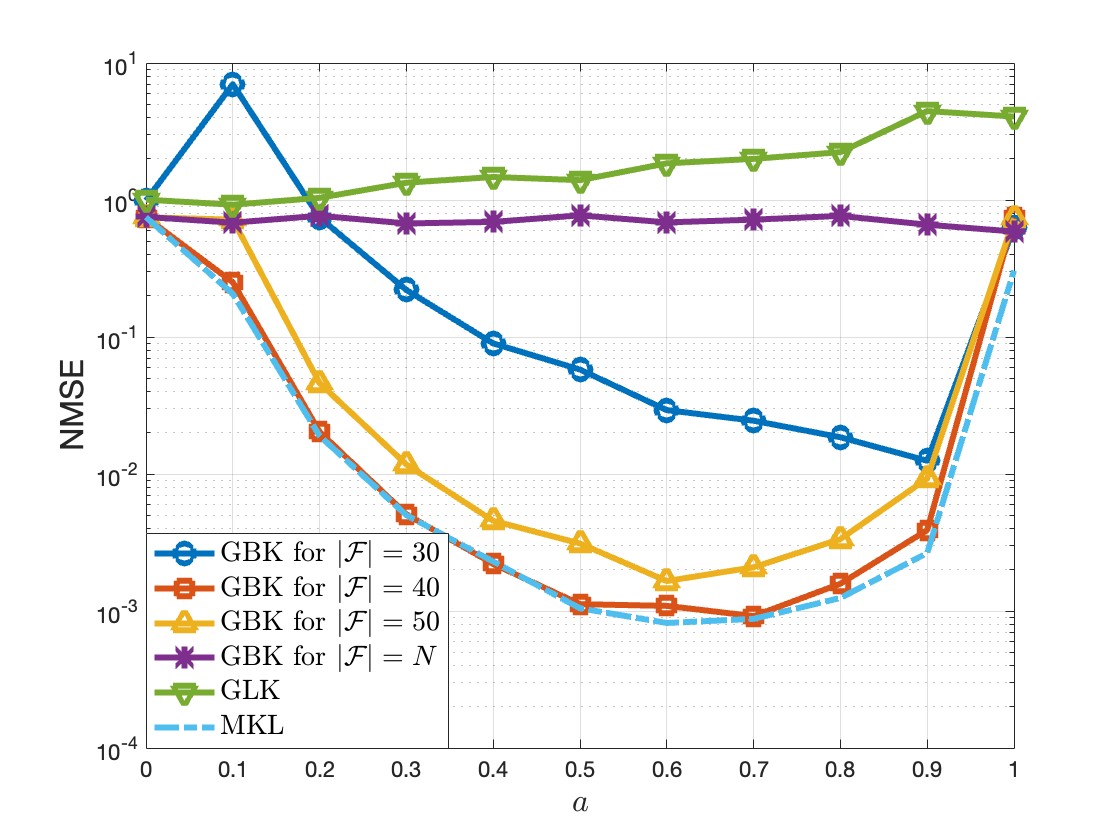}
					\parbox{3.8cm}{\tiny (c) NMSE variation with $a$ on Swiss roll graph.}
				\end{minipage}
				\begin{minipage}[t]{0.49\linewidth}
					\centering
					\includegraphics[width=\linewidth]{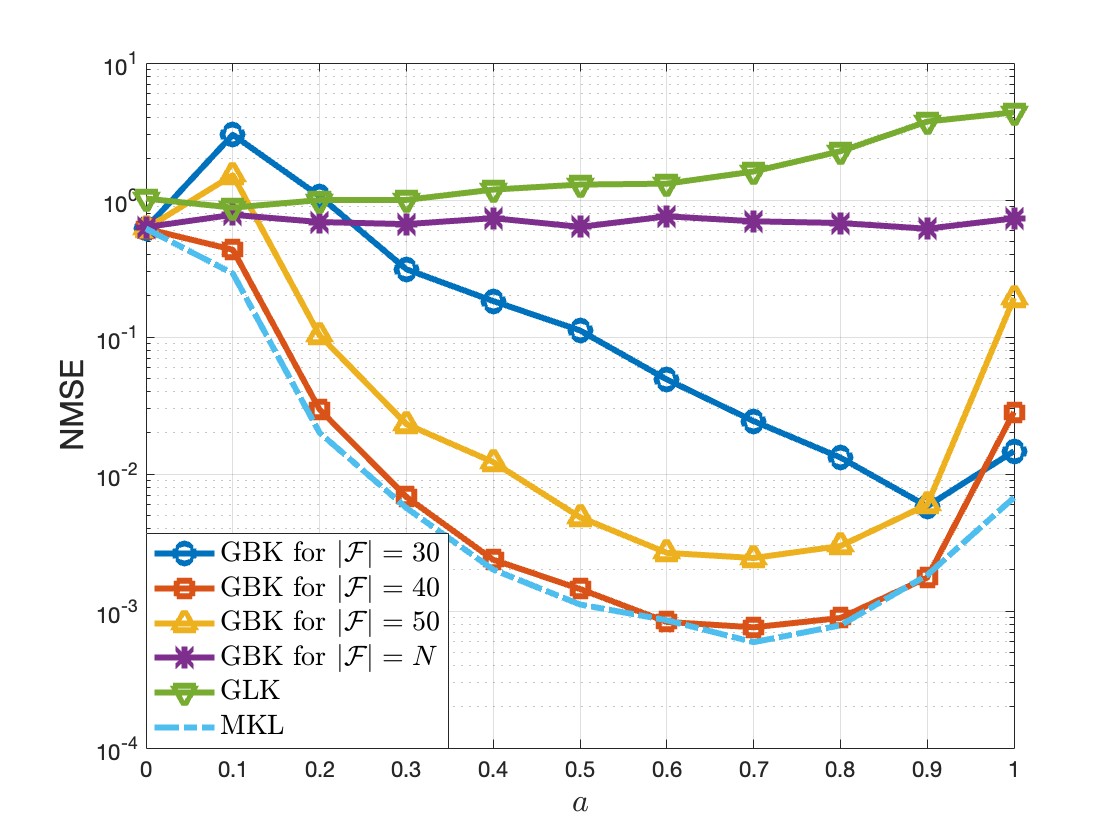}
					\parbox{3.8cm}{\tiny (d) NMSE variation with $a$ on community graph.}
				\end{minipage}
			\end{center}
			\caption{The NMSE between the reconstructed and original signals is evaluated across GBK methods.}
			\label{fig04}
		\end{figure}

		\subsubsection{Effects of Kernel Selection over Complex Manifolds}
		To effectively address the real data in the next subsection, we first simulate radar chirp signals to model sea clutter, characterized by linearly frequency-modulated (LFM) complex signals. Specifically, we generate a sea clutter graph signal with $N = 200$ points, where the frequency variation of the chirp signal is used to simulate the clutter. Three complex LFM signals, each with a different phase, represent sea clutter in different directions. The generated data consists of three-dimensional coordinates $x$, $y$, and $z$, each corresponding to a distinct signal. The first signal (for the $x$-coordinate) is given by: $\exp \left( 2\pi \mathrm{i}\left( \frac{0.1}{N} t^{2}\right)  \right)$, generating a periodic complex signal. The $y$- and $z$-coordinates are derived by introducing phase shifts of $ \pi / 3 $ and $ 2\pi / 3 $ to the $x$-coordinate signal, respectively.

		Thus, we treat this three-dimensional dataset as features in a complex manifold, with the first signal ($x$-coordinate) serving as the target non-bandlimited graph signal. Subsequently, both traditional and bandlimited kernel methods are applied to reconstruct the signal. The parameters for the traditional kernel are set to $\sigma = 0.5$, $c = 1$, and $d = 2$, with the complex manifold equipped with the Kähler metric. For the bandlimited kernel, we set $a = 0.7$ and $\mu = 10^{-4}$, with the regularization parameter $\gamma = 0.01$ for both methods. The bandlimited signal is randomly sampled with a Gaussian noise $\mathcal{N}(0, 0.01^2)$. For multi-kernel methods, we consider MKL on HGK with $\sigma = 0.5, 1, 2$, as well as MKL on all GBK.
		
		Fig. \ref{fig05} illustrates the relationship between NMSE and sampling rate for different kernels, with results averaged over 20 independent trials due to random sampling. The results indicate that most traditional kernels outperform the bandlimited kernels, with MKL based on HGK showing the best performance, followed by the single kernel HGK with $\sigma = 0.5$. For GBK, due to the non-bandlimited nature of the signal, reconstruction performance is less significant, with MKL only offering slight improvement.
		\begin{figure}[htbp]
			\begin{center}
				\begin{minipage}[t]{0.85\linewidth}
					\centering
					\includegraphics[width=\linewidth]{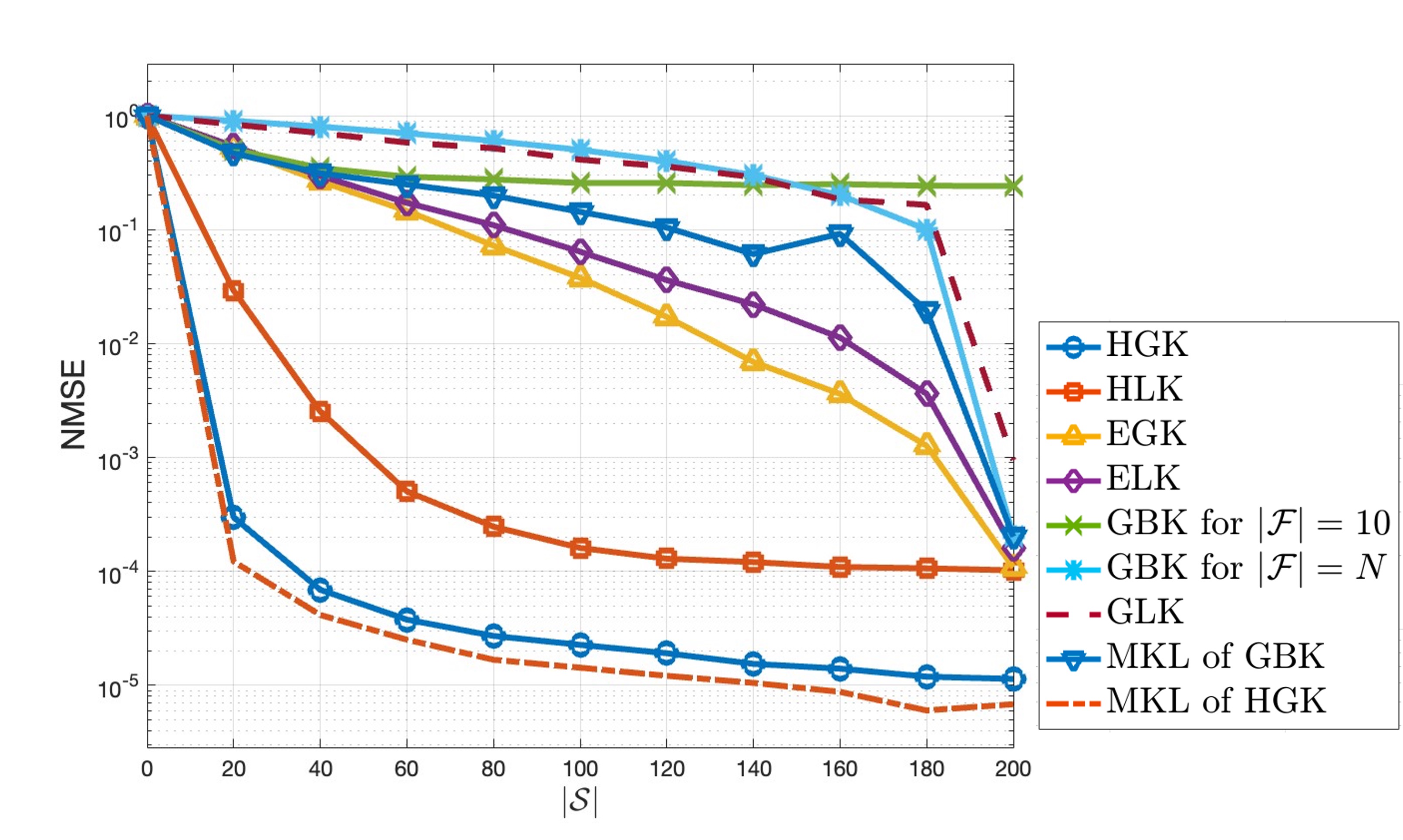}
				\end{minipage}
			\end{center}
			\caption{NMSE between the reconstructed and original signals of a simulated clutter radar.}
			\label{fig05}
		\end{figure}
		
		
		\subsection{Application to Real Data}\label{sec6.2}
		This section presents the experimental results based on real data, demonstrating the effectiveness of the proposed framework in reconstructing complex-valued graph signals. The experiments are conducted using a real-world dataset, focusing on signal reconstruction, performance evaluation, and comparison with conventional kernel-based methods.
		
		The dataset\footnote{Available online: \url{http://soma.mcmaster.ca/ipix/dartmouth/cdf051_100.html}} was obtained from a real radar system and contains complex-valued signal measurements across various azimuth angles and range bins. To ensure proper data formatting for reconstruction, preprocessing steps such as data normalization and correction are performed. Normalization is achieved by dividing the signal amplitudes of each range bin by their maximum value, while negative values are corrected by adding 256 to ensure positivity.
		
		We select the first $10^{4}$ of the graph signal and set $|\mathcal{S}|=2000$. Signal reconstruction is carried out using kernel-based methods on the complex graph manifold. Various kernel types are studied, and the signal corresponding to a specific range bin is reconstructed using the selected kernel and a predefined number of random sampling points. The reconstruction quality is evaluated using the NMSE, as shown in Table \ref{tab2}, with NMSE values under different kernel parameters for MKL applied to all traditional kernels.
		
		\begin{table*}[htbp]
			\centering
			\caption{Comparison of NMSE in Graph Signal Reconstruction Using Four Traditional Kernels under Different Kernel Parameters}\label{tab2}
			\footnotesize
			\begin{tabular}{cccccc}
				\toprule
				$\sigma$ & HGK&HLK & EGK & ELK & MKL\\
				\midrule
				$0.25$ & $5.1244 \times 10^{-4}$ & $4.71 \times 10^{-3}$& $9.137 \times 10^{-4}$&$5.32 \times 10^{-3}$ & $5.02 \times 10^{-4}$\vspace{0.1cm}\\
				$0.5$   &  $2.31 \times 10^{-4}$& $1.85 \times 10^{-3}$ & $2.40 \times 10^{-4}$ & $2.04 \times 10^{-3}$ & $1.61 \times 10^{-4}$ \vspace{0.1cm}\\
				$1$ &  $7.16 \times 10^{-4}$&$1.19 \times 10^{-3}$ & $3.77 \times 10^{-4}$ & $1.31 \times 10^{-3}$ & $3.24 \times 10^{-4}$\\
				\bottomrule
			\end{tabular}
		\end{table*}
		
		After reconstruction, the statistical properties of the optimal reconstructed signal are analyzed by comparing the probability density function (PDF) and cumulative distribution function (CDF) of the original and reconstructed signals. This provides insights in the signal's distribution and overall behavior, with theoretical distributions such as Rayleigh, Weibull, and K being used for further evaluation of the reconstruction’s alignment with expected statistical behavior.
		
		As shown in Fig. \ref{fig06}, the results indicate that the proposed methods achieve excellent signal reconstruction performance. The CDF and PDF of the reconstructed signal closely align with those of the original, and comparisons with theoretical distributions demonstrate strong consistency in statistical behavior. The performance across different kernels is evaluated, with the MKL for $\sigma=0.5$ yielding the best reconstruction quality, as reflected by the lowest NMSE. The analysis in Table \ref{tab2} further supports the methods' effectiveness in accurately reconstructing the signals.
		
		\begin{figure}[htbp]
			\begin{center}
				\begin{minipage}[t]{0.49\linewidth}
					\centering
					\includegraphics[width=\linewidth]{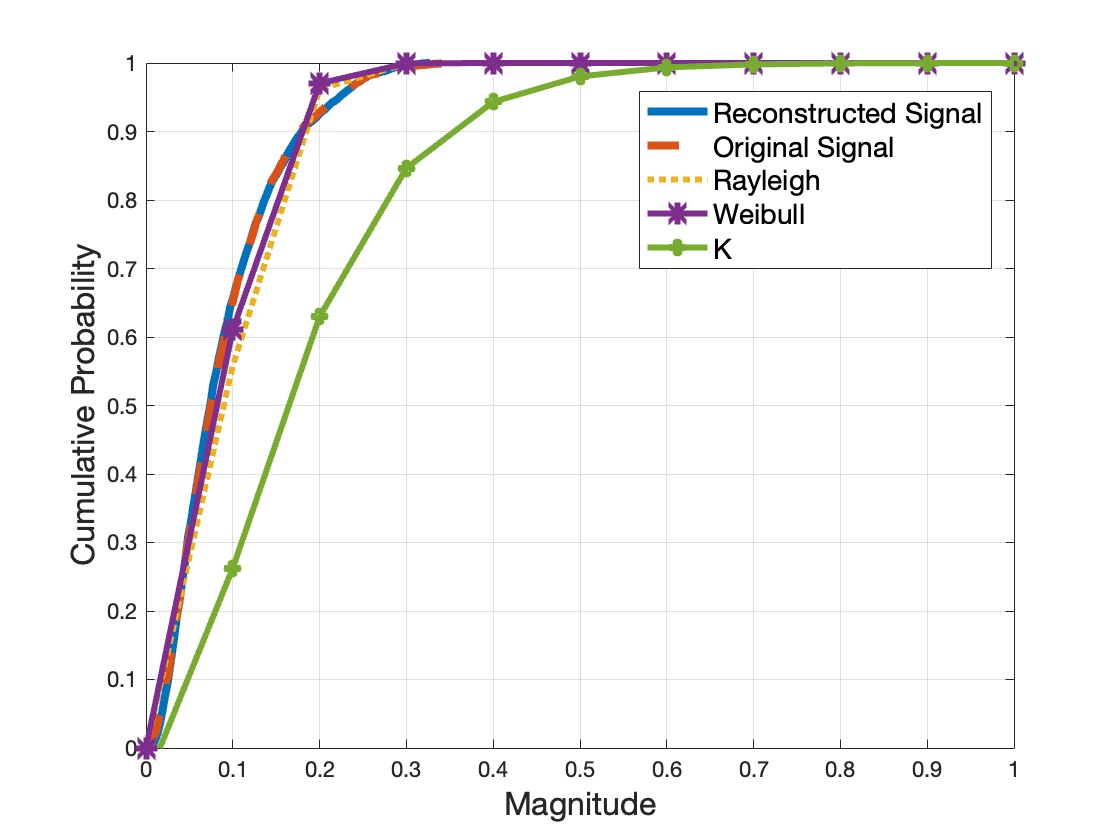}
					\parbox{2cm}{\tiny (a) CDF Comparison.}
				\end{minipage}
				\begin{minipage}[t]{0.49\linewidth}
					\centering
					\includegraphics[width=\linewidth]{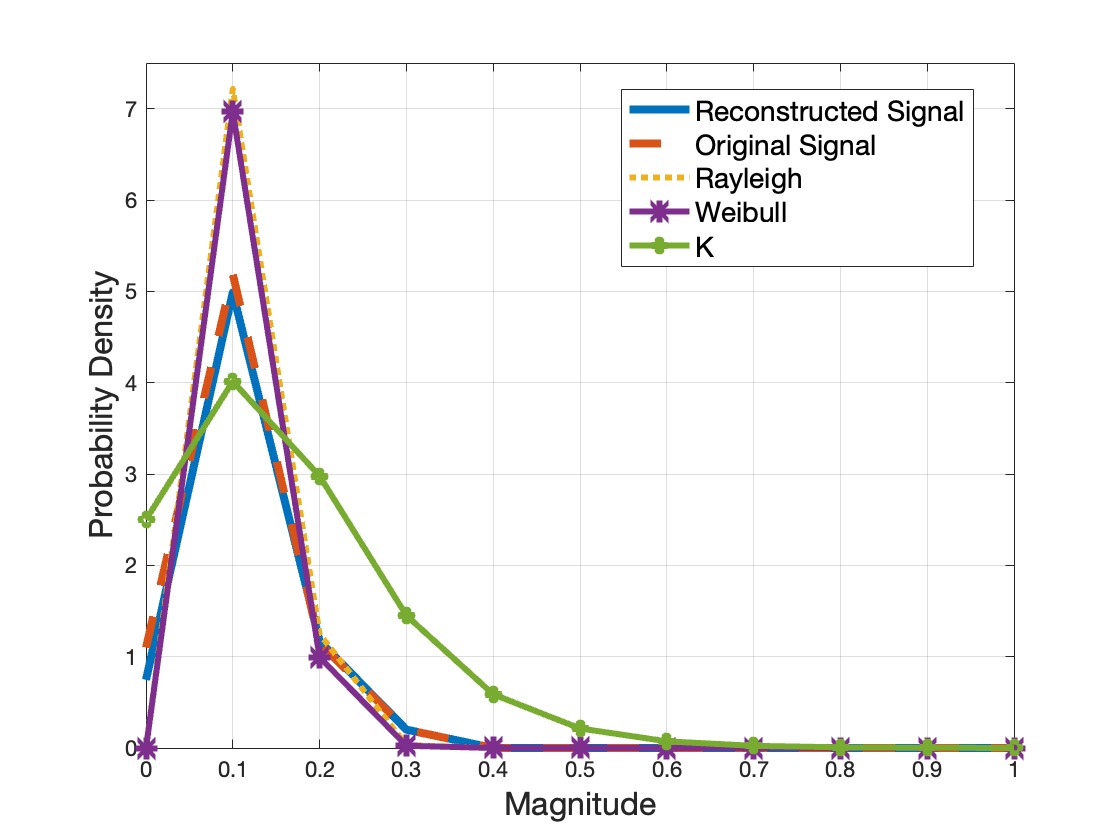}
					\parbox{2cm}{\tiny (b) PDF Comparison.}
				\end{minipage}
			\end{center}
			\caption{Comparison of reconstructed signals, original signals, and theoretical distributions (Rayleigh, Weibull, etc.) via PDF and CDF.}
			\label{fig06}
		\end{figure}
		

		\section{Conclusion}
		This paper proposed a novel framework for reconstructing complex-valued graph signals using kernel methods on complex manifolds, extending traditional graph signal processing. By embedding graph vertices into a higher-dimensional complex space, the framework approximated lower-dimensional manifolds and generalized reproducing kernel Hilbert spaces. Hermitian and geometric metrics were used to characterize the kernels and signals. Experiments on synthetic and real-world datasets demonstrated that the proposed approach achieved superior accuracy and efficiency compared to conventional kernel-based methods. Future research directions include the development of adaptive kernel design for complex manifolds, improved scalability for large graphs, and applications in areas such as quantum information, biomedical signal analysis, and social network modeling.
		
		
		\appendices
		\section{Hermitian Metric and Holomorphic Structure}
		\label{appendix A}
		In complex differential geometry, a Hermitian metric $h$ on a complex manifold $ \mathcal{M} $ is a smoothly varying Hermitian positive definite (HPD) form on the holomorphic tangent bundle $ T\mathcal{M} $. At each point $ z \in \mathcal{M} $, it defines a Hermitian inner product on the tangent space $ T_z\mathcal{M} $ as  
		\[
		h: T_z\mathcal{M} \times T_z\mathcal{M} \to \mathbb{C}.
		\]
		
		In local holomorphic coordinates $ (z^1,z^2, \dots, z^n) $, the Hermitian metric is expressed as  
		\[
		h = h_{i \bar{j}} \mathrm{d}z^i \otimes \mathrm{d}\bar{z}^j,
		\]
		where $ h_{i \bar{j}} $ is  an HPD matrix representing the local metric tensor and $ i, j \in \{1, 2, \dots, n\} $ are coordinate indices. This tensor governs the inner product of tangent vectors, i.e., 
		\[
		h(\partial_{z^i}, \partial_{z^j}) = h_{i \bar{j}}.
		\]
		
		For tangent vectors $ v = v^i \partial_{z^i} $ and $ w = w^j \partial_{z^j} $, the metric evaluates as  
		\[
		h(v, w) = \sum_{i, j} h_{i \bar{j}} v^i \overline{w}^j.
		\]
		The Hermitian metric $ h $ and  the corresponding HPD matrix $h_{i \bar{j}}$ encode key geometric properties such as distances, angles, and curvature, playing a crucial role in manifold-based signal processing and kernel methods.
		
		Under the complex manifold hypothesis, holomorphic structures ensure the analyticity and smoothness of signals, which are critical for kernel-based signal reconstruction and processing. Functions $f(z)$ on complex manifolds satisfy the \textit{Cauchy-Riemann equations}  
		\[
		\frac{\partial f}{\partial \overline{z}} = 0.
		\]  
		This condition guarantees signal smoothness, preventing local irregularities, aligning with the smoothing objectives of graph signal reconstruction. Moreover, RKHS constructed via kernel functions further ensure analyticity and smoothness of complex signals.
		
		
		\section{Proof of The Representative Theorem \ref{thm1}}\label{appendix B}
		Define $\mathcal{H}_1 = \text{span} \{ \kappa(\cdot, z_i)\mid z_i \in \mathcal{S} \}$ as the subspace spanned by the kernel functions over the sampling points $\mathcal{S}$. Any $f \in \mathcal{H}$ can be decomposed as $f = f_1 + f^\perp$, where $ f_1 \in \mathcal{H}_1$ and $f^\perp \in \mathcal{H}_1^\perp$. The optimization depends only on $f_1$, as $f^\perp$ contributes nothing to the objective function \cite{RKHS}. 
		
		Since $f_1 \in \mathcal{H}_1$, it can be expressed as a linear combination of kernel functions, namely, 
		\[
		f_1(z) = \sum_{m=1 }^{|\mathcal{S}|} \beta_m \kappa(z, z_m),
		\]  
		where $\bm{\beta} = (\beta_1, \beta_2,\dots, \beta_{|\mathcal{S}|})^\top$ are the coefficients representing the weights of the sampled points. For a given set of sampled points $ \mathcal{S} $, the function value vector is defined as  
		\[
		\bm{f}_1 = (f_1(z_1), f_1(z_2), \dots, f_1(z_{|\mathcal{S}|}))^\top \in \mathbb{C}^{|\mathcal{S}|}.
		\]
		Using matrix notation, this can be rewritten as  
		\[
		\bm{f}_1 = \mathbf{D} \mathbf{K} \mathbf{D}^\top \bm{\beta},
		\]
		where $\mathbf{K}$ represents the kernel matrix, and $\mathbf{D}$ is the sampling matrix. Notably, the key to the validity of this equation lies in the relationship $\bm{\alpha} = \mathbf{D}^\top \bm{\beta}$. In this context, the coefficient vector $\bm{\alpha}$ is sparse. In other words, the transpose of the sampling matrix simply maps the lower-dimensional coefficient vector back to its original higher-dimensional space, filling the corresponding positions with values while setting all other elements to zero.
		
		We aim to minimize
		\[
		\frac{1}{|\mathcal{S}|}  \left| \left| \bm{y}_{\mathcal{S}} - \mathbf{D}\bm{f} \right|\right|^2 + \gamma \bm{f}^{\mathrm{H}} \mathbf{K}^{\dagger} \bm{f}.
		\]  
		Since the solution $\bm{f}_1$ depends only on $\bm{\beta}$, the problem reduces to optimizing $\bm{\beta}$. Solving this yields the optimal coefficients $\bm{\beta}$, which uniquely determine $\bm{f}_1$, giving the optimal solution
		\[
		\bm{f}^{\text{opt}} = \mathbf{D}\mathbf{K} \mathbf{D}^\top \bm{\beta}.
		\]

		\section{Proof of The Smoothness Equation (Proposition \ref{pro:smo})}\label{appendix C}
		This proof leverages the symmetry of the kernel function and the matrix power to represent smoothness of the right-hand side of \eqref{eq:flf}, which can be expanded as 
		\begin{equation*}
			\begin{aligned} &\frac{1}{2} \sum^{N}_{n=1} \sum^{N}_{m=1} w^{a}_{nm}\left( f\left( z_{n}\right)  -f\left( z_{m}\right)  \right)^{2}  \\ 
				&=\frac{1}{2} \sum^{N}_{n=1} \sum^{N}_{m=1} w^{a}_{nm}\left( \left| f\left( z_{n}\right)  \right|^{2}  -2f\left( z_{n}\right)  f\left( z_{m}\right)  +\left| f\left( z_{m}\right)  \right|^{2}  \right)   \\ 
				&=\sum^{N}_{n=1} (d^{\mathrm{deg}}_{n})^{a}\left| f\left( z_{n}\right)  \right|^{2}  -\sum^{N}_{n=1} \sum^{N}_{m=1} w^{a}_{nm}f\left( z_{n}\right)  f\left( z_{m}\right)  \\ 
				&=\bm{f}^{\mathrm{H}}(\mathbf{D}^{\mathrm{deg}})^{a}\bm{f}-\bm{f}^{\mathrm{H}}\mathbf{W}^{a}\bm{f}\\ 
				&=\bm{f}^{\mathrm{H}}\mathbf{L}^{a}\bm{f},\end{aligned}
		\end{equation*}
		where $(\mathbf{D}^{\mathrm{deg}})^{a}$ and $\mathbf{W}^{a}$ are the degree matrix and the weight matrix raised to the power of $a$, respectively.

		\section{Proof of Lemma \ref{lem1}}\label{appendix D}
		The proof is similar to that of \cite{GUncertainty}. Given that $\mathbf{P}\bm{f} = \bm{f}$ and $\mathbf{B}^{a}\bm{f} = \bm{f}$, we obtain
		\[
		\mathbf{B}^{a}\mathbf{PB}^{a}\bm{f} = \mathbf{B}^{a}\mathbf{P}\bm{f} = \mathbf{B}^{a}\bm{f} = \bm{f},
		\]
		which implies that $\lambda_{\max}(\mathbf{B}^{a}\mathbf{PB}^{a}) = 1$. Conversely, if $\mathbf{B}^{a}\mathbf{PB}^{a}\bm{f} = \bm{f}$, then $\mathbf{B}^{a}\mathbf{PB}^{a}\bm{f} = \mathbf{B}^{a}\bm{f}$, and since $(\mathbf{B}^{a})^2 = \mathbf{B}^{a}$, it follows that $\mathbf{B}^{a}\bm{f} = \bm{f}$.
		For vertex localization,  the Rayleigh--Ritz theorem yields
		\[
		\max_{\bm{f}} \frac{\bm{f}^{\ast} \mathbf{P} \bm{f}}{\bm{f}^{\ast} \bm{f}} = \max_{\bm{f}} \frac{\bm{f}^{\ast} \mathbf{B}^{-a} \mathbf{P} \mathbf{B}^{a} \bm{f}}{\bm{f}^{\ast} \bm{f}} = 1.
		\]
		Thus, $\mathbf{P}\bm{f} = \bm{f}$, confirming perfect localization in both the vertex and spectral domains.

		

		

		\newpage
		
		
		
		
		

		
	\end{document}